%% file: RegularityMeasureForCFG.tex
\documentclass[a4paper]{article}
\usepackage{acronym}
\usepackage{amsthm}
\usepackage{mathbbol}
\usepackage{amsfonts}
\usepackage{amssymb}
\usepackage{extarrows}
\usepackage{tikz}
\usepackage{enumerate}
\usepackage[asymmetric,left=0.75in,top=0.5in,right=0.75in,bottom=0.75in]{geometry}

\newtheorem{theorem}{Theorem}

\newtheorem{lemma}[theorem]{Lemma}

\newtheorem{proposition}[theorem]{Proposition}

\newtheorem{definition}[theorem]{Definition}

\begin{document}
\title{A Regularity Measure for Context Free Grammars}
\author{M. Praveen\\The Institute of Mathematical Sciences, Chennai, India}
\date{}
\maketitle
\input{macros}
\input{tikzmacros}

\begin{abstract}
  Parikh's theorem states that every \ac{cfl} has the same Parikh
  image as that of a regular language. A finite state automaton
  accepting such a regular language is called a Parikh-equivalent
  automaton. In the worst case, the number of states in any
  non-deterministic Parikh-equivalent automaton is exponentially large
  in the size of the \ac{CFG}. We associate a regularity width
  $\regms$ with a \ac{CFG} that measures the closeness of the \ac{cfl}
  with regular languages. The degree $\degree$ of a \ac{CFG} is one
  less than the maximum number of variable occurrences in the right
  hand side of any production. Given a \ac{CFG} with $\numvars$
  variables, we construct a Parikh-equivalent non-deterministic
  automaton whose number of states is upper bounded by a polynomial in
  $\numvars ( \regms^{2\regms ( \degree + 1)})$, the degree of the
  polynomial being a small fixed constant. Our procedure is
  constructive and runs in time polynomial in the size of the
  automaton. In the terminology of parameterized complexity, we prove
  that constructing a Parikh-equivalent automaton for a given \ac{CFG}
  is \ac{fpt} when the degree $\degree$ and regularity width $\regms$
  are parameters.  We also give an example from program
  verification domain where the degree and regularity are small
  compared to the size of the grammar.
\end{abstract}
\acresetall

\section{Introduction}
The Parikh image $\prkimg ( \word)$ of a word $\word$ over a finite
alphabet $\termins$ is a mapping $\prkimg ( \word): \termins \to \nat$
such that for each letter $\letter \in \termins$, $\prkimg ( \word)
( \letter)$ is the number of times $\letter$ occurs in $\word$. The
Parikh image of a language is the set of Parikh images of its words. The
well known Parikh's theorem \cite{RJP1966} states that for every
\ac{cfl} $\lang$, there is a regular language with the same Parikh
image as that of $\lang$. This fundamental result in automata theory
has many applications, including verification \cite{OHI1978, TL2010,
EG2011}, equational horn clauses \cite{VSS2005} and automata theory
itself \cite{B2007}.

Apart from the equivalence itself, the complexity of computing a
representation of the regular language or the Parikh image is crucial
to the efficiency of many applications. There are examples where any
Parikh-equivalent non-deterministic automaton is exponentially large
in the size of the \ac{CFG} (see \secref{sec:cfg}). However,
as is frequently the case, instances of this problem arising in
applications have some structure that can be exploited to compute
smaller automata. In this paper, we introduce a systematic way of
measuring the ``closeness'' of a given \ac{cfl} to regular languages
and show that closer the \ac{cfl} is to regular languages, smaller
will be the size of Parikh-equivalent non-deterministic automata. More
precisely,
\begin{enumerate}
  \item We define a number called regularity width $\regms$ that can
    be computed from a given \ac{CFG}. If the \ac{CFG} happens to be a
    regular grammar, then its regularity width will be $1$.
  \item As an illustration, we show that instances of \acp{CFG} arising
    from a verification application \cite{EG2011} will have small
    values of regularity width $\regms$.
  \item With $\numvars$ denoting the number of variables in the given
    \ac{CFG} and degree $\degree$ being one less than the maximum
    number of variable occurrences in the right hand side of any
    production, we show that a Parikh-equivalent non-deterministic
    finite automaton can be constructed whose number of states is
    upper bounded by a polynomial in $\numvars ( \regms^{2\regms (
    \degree + 1)})$, the degree of the polynomial being a small fixed
    constant.
\end{enumerate}

Finer study of automaton complexity has been done before,
e.g., \cite{EG2011, AWT2010}. In \cite{EG2011}, a parameter called
number of procedure variables $\numpv$ is introduced and it is proved
that when $\numpv$ is a fixed constant, a Parikh-equivalent
non-deterministic automaton can be constructed whose number of states
is polynomial in $\numvars$. The result in this paper is both a
generalization and refinement of the results in \cite{EG2011}. It is a
generalization since the regularity width $\regms$ defined here is
linear in the number of procedure variables $\numpv$ for \acp{CFG}
arising from problems being considered in \cite{EG2011}, and $\regms$
can be computed for any given \ac{CFG}. Our result is a refinement
since \cite{EG2011} gives automata sizes of the kind $\numvars ^{
\numpv}$ while we give automata sizes of the kind $\numvars (
\regms^{4\regms})$. To get a rough idea of the kind of difference this
can make asymptotically, consider the ratio between $\numvars^{\numpv
+ 1}$ and $2^{\numpv}\numvars$ taken from \cite{DFS1999}: for
$\numvars = 100$ and $\numpv = 10$, the ratio is $9.8 \times 10^{14}$
while for $\numvars = 150$ and $\numpv= 20$, the ratio is $2.1 \times
10^{35}$.

To systematically explore the possibility of finding efficient
algorithms for restricted cases of computationally hard problems,
Downey and Fellows introduced parameterized complexity \cite{DF1999}.
If $\numvars$ is the size of an input instance and $\regms$ is its
parameter (that is usually much lesser than $\numvars$), then
algorithms with running time $\Oh ( \numvars ^{\poly ( \regms)})$ are
called XP algorithms and those with running time $\Oh (
f(\regms))\poly ( \numvars)$ are called \ac{fpt} algorithms. Here,
$\poly$ is any polynomial and $f$ is any computable function (usually
required to be single exponential or less to be of any immediate
practical use). It is known that the class of XP algorithms is
strictly more powerful than the class of \ac{fpt} algorithms. The
procedure we give for constructing the automaton runs in time
polynomial in the size of the output and hence is \ac{fpt}.

Our results build on technique based on pumping for \acp{cfl}
\cite{RJP1966,JG1977}. Additional techniques and arguments are needed
to closely control where and how pumping is done so that for
\acp{CFG} with small regularity width, smaller automata suffice.
Techniques from graph theory and tree decompositions \cite[Chapter
6]{DF1999} are used in arguments on size of the constructed automata.

\textbf{Related work:} A finer analysis of the automata size has been
done in \cite{AWT2010} with focus on the size of the alphabet. In
\cite{EGKL2010}, the size of automata are related to finite index
\acp{CFG}. Some of the techniques here are inspired by insights given
in \cite{EGKL2010}. Since Parikh images of \acp{cfl} are semilinear,
they can be represented in Presburger arithmetic. In \cite{VSS2005},
the complexity of Presburger formula representation has been
considered. The algebraic view of Parikh's theorem has been studied in
\cite{DLP1973,HK1999}.

\section{Preliminaries}
\subsection{\aclp{CFG}}
\label{sec:cfg}
Let $\termins$ be a finite set of symbols, called \emph{terminals}. A
word $\word$ over $\termins$ is any finite sequence of terminals. The
empty sequence is denoted by $\epsilon$. The word obtained by
concatenating $\idxo \in \nat$ copies of $\word$ is denoted by
$\word^{\idxo}$. The set of all words over $\termins$ is denoted
$\termins^{*}$. A language $\lang$ is any subset of $\termins^{*}$.
The \emph{Parikh image} $\prkimg( \word)$ of a word is a mapping
$\prkimg( \word): \termins \to \nat$ such that for each $\letter \in
\termins$, $\prkimg( \word)(\letter)$ is the number of times $\letter$
occurs in $\word$. The Parikh image of $\epsilon$ is denoted by
$\zv$. The Parikh image of a language $\lang \subseteq
\termins^{*}$ is the set of mappings $\{\prkimg (\word ) \mid \word
\in \lang\}$.

We follow the notation of \cite[Chapter 5]{HMU2007}. Let $\cfg = (
\vars, \termins, \prodns, \macaxiom)$ be a \ac{CFG} with a set
$\vars=\{\var_{1}, \dots, \var_{\numvars}\}$ of \emph{variables}, a
finite set $\termins$ of \emph{terminals}, a finite set $\prodns \subseteq
\vars \times (\vars \cup \termins)^{*}$ of \emph{productions} and an
\emph{axiom} $\macaxiom \in \vars$. We denote words over $\termins$ by
$\word$, $\word_{1}$ etc. A production $(\var, \word_{0}
\var_{1} \word_{1} \cdots \var_{\prlen} \word_{\prlen}) \in \prodns$
is denoted as $\var \rightsquigarrow \word_{0} \var_{1} \word_{1} \cdots
\var_{\prlen} \word_{\prlen}$. The \emph{degree} $\degree$ of $\cfg$ is
defined to be $-1 + \max\{\prlen \mid \var \rightsquigarrow
\word_{0} \var_{1} \word_{1} \cdots \var_{\prlen} \word_{\prlen}
\text{ is a production}\}$.

The set of words $\lang(\cfg)$ over $\termins$ generated by $\cfg$ is
called the language of $\cfg$ and is called a \ac{cfl}.  Parikh's
theorem \cite{RJP1966} states that the Parikh image of every
\ac{cfl} is equal to that of a regular language. Given a \ac{CFG}
$\cfg$, a finite automaton accepting a language whose Parikh image is
the same as that of $\lang ( \cfg)$ is called a Parikh-equivalent
automaton for $\cfg$. Consider the \ac{CFG} $\cfg_{n}$ with
productions $\{\var_{\idxt} \rightsquigarrow \var_{\idxt - 1}
\var_{ \idxt - 1} \mid 2 \le \idxt \le n\} \cup \{\var_{1}
\rightsquigarrow a\}$ and axiom $\macaxiom = \var_{n}$. The language
of $\cfg_{n}$ is the singleton set $\{a ^{2^{n-1}}\}$ and hence the
smallest Parikh-equivalent non-deterministic automaton has
$2^{n-1} + 1$ states.

We assume familiarity with parse trees \cite[Chapter 5]{HMU2007} and
yield of parse trees. We denote parse trees by $\prstr$, $\prstr_{1}$
etc.~and their yeilds by $\yield(\prstr)$, $\yield( \prstr_{1})$ etc.
The variable labelling the root of a parse tree $\prstr$ is denoted by
$\troot(\prstr)$. For any variable $\var \in \vars$, the parse tree
$\prstr$ is defined to be \emph{$\var$-recurrence free} if in any path
from the root to a leaf of $\prstr$, $\var$ occurs at most once. We
define $\prstr$ to be \emph{$\var$-occurrence free} if $\var$ does not
occur anywhere in $\prstr$. If $\prstr$ is $\var$-recurrence free and
rooted at $\var$, then any proper subtree of $\prstr$ is
$\var$-occurrence free. We write $\prstr = \prstr_{1} \cdot
\prstr_{2}$ to denote that $\prstr_{1}$ is a parse tree except that
exactly one leaf $\tn$ is labelled by a variable, say $\var$, instead
of a terminal; the tree $\prstr_{2}$ is a parse tree rooted at $\var$;
and the parse tree $\prstr$ is obtained from $\prstr_{1}$ by replacing
the leaf $\tn$ with $\prstr_{2}$. The height of a parse tree
$\prstr$ is denoted by $\height ( \prstr)$.

Using the fact that a Parikh-equivalent automaton need not preserve
the order in which letters occur in accepted words, the following
lemma allows us to manipulate parse trees so that the automaton need
not keep track of too many occurrences of the same variable. This is
also one of the key observations used in the automaton construction
given in \cite{EGKL2010}.

\begin{lemma}[\cite{EGKL2010}]
  \label{lem:parseTreeReduceRecurrence}
  Suppose  $\prstr_{1}, \prstr_{2}$ are two parse trees and $\var \in
  \vars$ is a variable such that $\prstr_{1}$ is not $\var$-recurrence
  free and $\prstr_{2}$ is not $\var$-occurrence free. Then there are
  parse trees $\prstr_{1}', \prstr_{2}'$ such that
  \begin{enumerate}
    \item $\troot(\prstr_{1}') = \troot( \prstr_{1})$,
      $\troot(\prstr_{2}') = \troot( \prstr_{2})$ and 
    \item $\prkimg( \yield( \prstr_{1})) + \prkimg( \yield(
      \prstr_{2})) = \prkimg( \yield( \prstr_{1}')) + \prkimg( \yield(
      \prstr_{2}'))$ and
    \item $\prstr_{1}$ is $\var$-recurrence free.
  \end{enumerate}
\end{lemma}
\begin{proof}
  Since there is a path from the root to a leaf of $\prstr_{1}$ in
  which $\var$ occurs at least twice, we can write $\prstr_{1} =
  \prstr\cdot \prstr' \cdot \prstr''$ where $\prstr'$ and $\prstr''$
  are rooted at $\var$. The parse tree $\prstr_{2}$ can similarly be
  written as $\prstr_{2} = \prstr_{2}'' \cdot \prstr_{2}'''$, where
  $\prstr_{2}'''$ is rooted at $\var$. Replace $\prstr_{1}$ by $\prstr
  \cdot \prstr''$ and $\prstr_{2}$ by $\prstr_{2}'' \cdot \prstr' \cdot
  \prstr_{2}'''$ (i.e., remove the subtree $\prstr'$ from $\prstr_{1}$
  and insert it into $\prstr_{2}$). This will reduce the number of
  nodes in $\prstr_{1}$. Repeat this process until $\prstr_{1}$ is
  $\var$-recurrence free. This process will terminate after finitely
  many steps since there are only finitely many nodes in $\prstr_{1}$
  to begin with.
\end{proof}

\subsection{Graphs and Tree Decompositions}
\begin{definition}[Tree decomposition, treewidth]
A \emph{tree decomposition} of an undirected graph $\graph =
(\vertexs, \edges)$ is a pair
$(\tree,(\bag_{\tn})_{\tn\in \nodes(\tree)})$,
where $\tree$ is a tree and $(\bag_{\tn})_{\tn\in
\nodes(\tree)}$ is a family of subsets of $\vertexs$ (called
\emph{bags}) such that:
\begin{itemize}
  \item For all $\verto\in \vertexs$, the set $\{\tn\in
    \nodes(\tree)\mid \verto\in \bag_{\tn}\}$ is
    nonempty and connected in $\tree$.
  \item For every edge $(\verto_{1},\verto_{2})\in \edges$,
    there is a $\tn\in \nodes(\tree)$ such that
    $\verto_{1},\verto_{2} \in \bag_{\tn}$.
\end{itemize}
The width of such a decomposition is the number
$\max\{|\bag_{\tn}|\mid \tn\in \nodes(\tree)\}-1$. The
\emph{treewidth} $\tw(\graph)$ of $\graph$ is the minimum of
the widths of all tree decompositions of $\graph$. 
\end{definition}

If a graph $\graph$ has an edge between every pair of vertices among
$\{\verto_{1}, \dots, \verto_{\prlen}\}$, then $\{\verto_{1}, \dots,
\verto_{\prlen}\}$ is said to \emph{induce a clique} in $\graph$.

\begin{lemma}[{\cite[Lemma 6.49]{DF1999}}]
  \label{lem:cliqueTreeDecomp}
  If the set of vertices $\{\verto_{1}, \dots, \verto_{\prlen}\}$
  induce a clique in the graph $\graph$, then every tree decomposition
  of $\graph$ will have a bag $\bag$ with $\{\verto_{1}, \dots,
  \verto_{\prlen}\} \subseteq \bag$.
\end{lemma}

\section{The Regularity Measure}
\label{sec:regms}
In this section, we define the regularity width and give an example
application where regularity width is much lower than the size of the
\ac{CFG}.

We define a binary accessibility relation $\rightarrow$ between
variables in $\vars$ as follows. We have $\var \rightarrow \var'$ if
there is a production $\var \rightsquigarrow (\vars \cup
\termins)^{*} \var' (\vars \cup \termins)^{*}$. The reachability
relation $\reachrel$ is the transitive closure of the accessibility
relation $\rightarrow$.
\begin{definition}[Reminder graph and regularity width]
  \label{def:reminderGraph}
  For a \ac{CFG} $\cfg=( \vars, \termins, \prodns, \macaxiom)$, its
  \textbf{reminder graph} $\remgr{\cfg}$ is a graph whose set of
  vertices is $\vars$ and set of edges $\edges$ is as follows. For every
  production $\var \rightsquigarrow \word_{0} \var_{1} \word_{1} \cdots
  \var_{\prlen} \word_{\prlen}$ in $\cfg$ with $\prlen \ge 2$ and
  every $1\le \idxo, \idxt \le \prlen$, following edges are present:
  \begin{align}
    \idxo \ne \idxt \text{ implies } (\var_{\idxo}, \var_{\idxt}) \in
    \edges
    \label{eq:siblingEdge}
  \end{align}
  \begin{align}
    \forall \var' \in \vars \setminus \{\var_{\idxt}\},\quad
    \var_{\idxo} \reachrel \var' \text{ implies } (\var',
    \var_{\idxt}) \in \edges
    \label{eq:descendandEdge}
  \end{align}
  The \textbf{regularity width} $\regms$ of $\cfg$ is defined to be
  $\tw(\remgr{\cfg}) + 1$.
\end{definition}

The intuition behind the definition of reminder graph is explained in
the following diagram of a parse tree.
\begin{figure}[!htp]
  \begin{center}
    \input{ReminderGraphExample}
  \end{center}
  \caption{An illustration for reminder graph}
  \label{fig:reminderGraphIllustration}
\end{figure}
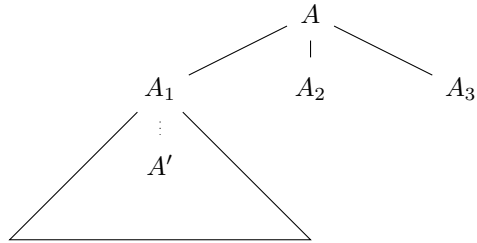
Suppose we are trying to imitate this parse tree through a finite
state automaton and we go down the tree rooted at
$\var_{1}$. The edge between $\var_{2}$ and $\var_{3}$ reminds us
that both $\var_{2}$ and $\var_{3}$ have to be followed up later.
Going down the production starting from $\var_{1}$, suppose we reach
the variable $\var'$. The edge between $\var'$ and $\var_{2}$ reminds
us that $\var_{2}$ is yet to be followed up. The edge between
$\var'$ and $\var_{1}$ reminds us that we have already gone down a
tree rooted at $\var_{1}$, so we should avoid going down subtrees that
are also rooted at $\var_{1}$, thus avoiding the necessity to keep
track of too many $\var_{1}$s.

Since regular grammars have at most one variable in the right hand
side of any production, their reminder graphs do not have any edges.
Hence, regularity width of regular grammars is $1$. Following is a
more interesting example from \cite{EG2011}: if there are programs
running in many threads in parallel and synchronizing on common
actions, certain verification problem reduces to reasoning about
Parikh images of \ac{cfl} derived from the programs. If there is a
program point $C_{1}$ from which some action $a$ can be performed and
program point $C_{2}$ can be reached, then we create a production
$C_{1} \rightsquigarrow a C_{2}$. If a subroutine $P_{0}$ is invoked
at point $C_{1}$ and then $C_{2}$ is reached, we create a production
$C_{1} \rightsquigarrow P_{0}C_{2}$. Lets call variables like $P_{0}$ and
$C_{2}$ ports. The number of ports is twice the number of control
locations that invoke some subroutine. The only edges in the reminder
graph are those between ports and other variables. This gives an easy
way to construct a tree decomposition of the reminder graph: if
$C_{1}, \dots, C_{\prlen}$ is a list of all program points, then
create a path with $\prlen$ nodes, each node $\tn_{\idxo}$ associated
with a bag $\bag_{\idxo}$, $1 \le \idxo \le \prlen$. If $\mathcal{P}$
is the set of all ports, then setting $\bag_{\idxo} = \mathcal{P} \cup
\{C_{\idxo}\}$ for every $1 \le \idxo \le \prlen$ will give us a tree
decomposition of the reminder graph with each bag containing
$|\mathcal{P}| + 1$ elements. Hence, for a program with
$|\mathcal{P}|$ ports, the associated \ac{CFG} has regularity width
$|\mathcal{P}| + 1$.

\section{The Automaton Construction}
Let $\macset$ be any set. A \emph{multiset} $\mults{\elem}$ over
$\macset$ is a mapping $\mults{\elem}: \macset \to \nat$.  We
sometimes use the notation $\Lbrack \elem_{1}, \elem_{3}, \elem_{3}
\Rbrack$ to denote the multiset that maps $1$ to $\elem_{1}$, $2$ to
$\elem_{3}$ and $0$ to all others. The empty multiset is denoted
$\emptyset$. Given two multisets $\mults{\elem}_{1}$ and
$\mults{\elem}_{2}$, their sum $\mults{\elem}_{1} \oplus
\mults{\elem}_{2}$ is defined to be the mapping such that
$\mults{\elem}_{1} \oplus \mults{\elem}_{2} ( \elem) =
\mults{\elem}_{1} ( \elem) + \mults{\elem}_{2} (\elem)$ for all $\elem
\in \macset$. If $\mults{\elem} ( \elem) \ge 1$, then $\mults{\elem}
\ominus \Lbrack \elem \Rbrack$ is defined to be the mapping such that
$\mults{\elem} \ominus \Lbrack \elem \Rbrack ( \elem') = \mults{\elem}
( \elem')$ for all $\elem' \in \macset \setminus \{\elem\}$ and
$\mults{\elem} \ominus \Lbrack \elem \Rbrack (\elem) = \mults{\elem} (
\elem) -1$.

The automaton we construct will have as their states sequences of
reminder pairs defined below.
\begin{definition}
  \label{def:reminderPairSeq}
  A \textbf{reminder pair} $\rp$ is a tuple $( \var, \varmults)$ where
  $\var \in \vars \cup \{\bot\}$ is a variable or a special symbol
  $\bot$ and $\varmults: \vars \to \nat$ is a
  multiset over $\vars$. The first component of this tuple will be
  referred to as $\rp.\cur$ and the second component as
  $\rp.\folup$. A \textbf{reminder sequence} $\rs = \rp_{1} \cdots
  \rp_{\rslen}$ is a sequence of reminder pairs, where $\rslen =
  |\rs|$ is the length of $\rs$. The reminder sequence $\rs = \rp_{1}
  \cdots \rp_{\rslen}$ \textbf{ends with $\var$} if $\rp_{\rslen}.\cur
  = \var$. The variable $\var$ \textbf{occurs $\idxo$ times in $\rs$}
  if $|\{\idxt \in \nat \mid 1\le \idxt \le \rslen,\rp_{\idxt}.\cur =
  \var\}| = \idxo$.
\end{definition}

Given a \ac{CFG} $\cfg$, we define a finite state automaton $\aut(\cfg)$
by describing its states and transition relation below.
\begin{definition}
  \label{def:ReminderAut}
  Let $\cfg$ be a \ac{CFG}. Then $\aut( \cfg)$ is a finite state
  automaton whose initial state is the reminder sequence consisting of
  the single reminder pair $(\macaxiom, \emptyset)$, where $\macaxiom$
  is the axiom of $\cfg$. A reminder sequence $\rs$ is a state of
  $\aut( \cfg)$ if for any variable $\var \in \vars$, $\var$ occurs at
  most twice in $\rs$ and it is reachable from $(\macaxiom,
  \emptyset)$ by the transition relation $\xLongrightarrow{}$
  specified below. In the following, $\rs$ could be the empty sequence
  $\epsilon$ too.
  \begin{enumerate}
    \item \label{it:eext} If $\rs\cdot (\var, \emptyset)$ is a
      reminder sequence and $\var \rightsquigarrow \word_{0} \var_{1}
      \word_{1} \cdots \var_{\prlen} \word_{\prlen}$ is a production
      with $\prlen \ge 1$, then $\rs\cdot (\var, \emptyset)
      \xLongrightarrow{ \word_{0} \word_{1} \cdots \word_{\prlen}} \rs
      \cdot (\var_{\idxt}, \Lbrack \var_{1}, \dots, \var_{\prlen}
      \Rbrack \ominus \Lbrack \var_{\idxt} \Rbrack)$ for every $1\le
      \idxt \le \prlen$.
    \item \label{it:neext} If $\rs\cdot (\var, \varmults)$ is a reminder
      sequence, $\varmults \ne \emptyset$ and $\var \rightsquigarrow
      \word_{0} \var_{1} \word_{1} \cdots \var_{\prlen}
      \word_{\prlen}$ is a production with $\prlen \ge 1$, then
      $\rs\cdot (\var, \varmults) \xLongrightarrow{ \word_{0} \word_{1}
      \cdots \word_{\prlen}} \rs \cdot ( \var, \varmults) \cdot
      (\var_{\idxt}, \Lbrack \var_{1}, \dots, \var_{\prlen} \Rbrack
      \ominus \Lbrack \var_{\idxt} \Rbrack)$ for every $1\le \idxt \le
      \prlen$.
    \item \label{it:erepl} If $\var \rightsquigarrow \word$ is a
      production, then $\rs \cdot (\var, \emptyset)
      \xLongrightarrow{\word} \rs \cdot (\bot, \emptyset)$.
    \item \label{it:nerepl} If $\var \rightsquigarrow \word$ is a
      production, then $\rs \cdot (\var, \varmults \oplus \Lbrack \var'
      \Rbrack) \xLongrightarrow{\word} \rs \cdot (\var', \varmults)$ for
      every $\var' \in \vars$.
    \item \label{it:shorten} $\rs \cdot (\var, \varmults \oplus \Lbrack
      \var' \Rbrack) \cdot (\bot, \emptyset)
      \xLongrightarrow{\epsilon} \rs \cdot ( \var', \varmults)$ for every
      $\var' \in \vars$.
  \end{enumerate}
  The final state of $\aut( \cfg)$ is $(\bot, \emptyset)$.
\end{definition}

If $\rs$ is a reminder sequence, we denote by $\rs\{\idxo\} = \{\var
\in \vars \mid \rp_{\idxo}.\folup \oplus \Lbrack \rp_{\idxo}.\cur
\Rbrack (\var) \ge 1\}$ the set of those variables that occur in the
$\idxo$\textsuperscript{th} reminder pair of $\rs$. The following
lemma is the motivation for using the treewidth of the reminder graph
to define regularity width.

\begin{lemma}
  \label{lem:autStateRemGraphClique}
  Let $\cfg$ be a \ac{CFG} with degree $\degree$ and $\rs$ be a state
  of $\aut ( \cfg)$.  The set of variables $\bigcup_{1\le \idxo \le
  |\rs|}\rs\{ \idxo\}$ induces a clique in the reminder graph of
  $\cfg$. For any $1\le \idxo \le |\rs|$, $\sum_{\var \in
  \vars}\rp_{\idxo}.\folup(\var) \le \degree$.
\end{lemma}
\begin{proof}
  Let $\rs$ be any state of $\aut( \cfg)$ and $1 \le \idxo , \idxt \le
  |\rs|$ be positions of $\rs$. By induction on the minimum number
  $\trlen$ of transition relation pairs $\rs_{1}
  \xLongrightarrow{} \rs_{2}$ that have to be traversed to reach $\rs$
  from $(\macaxiom, \emptyset)$, we will prove the following claims:
  \begin{enumerate}[I]
    \item \label{it:prdnPresent} If $\rp_{\idxo}.\folup \ne
      \emptyset$, then there is some production $\var \rightsquigarrow
      \word_{0} \var_{1} \word_{1} \cdots \var_{\prlen}
      \word_{\prlen}$ such that $\rs\{\idxo\} \subseteq \{\var_{1},
      \dots, \var_{\prlen}\}$ and $\prlen \ge 2$.
    \item \label{it:noSucc} If $\rp_{\idxo}.\folup = \emptyset$, then
      $(\rp_{\idxo}.\cur, \emptyset)$ is the last reminder pair of
      $\rs$.
    \item \label{it:descReach} If $\idxo < \idxt$, then for any
      variable $\var \in \rs\{\idxt\}$, $\rp_{\idxo}.\cur \reachrel
      \var$.
    \item \label{it:numMembFolup} $\sum_{\var \in
      \vars}\rp_{\idxo}.\folup(\var) \le \degree$.
  \end{enumerate}
  In the following, it is clearly seen that claim
  \ref{it:numMembFolup} holds, so it will not be mentioned explicitly.

  \emph{Base case $\trlen = 0$:} Here, $\rs = (\macaxiom, \emptyset)$
  for which all the claims clearly hold.

  \emph{Induction step:} We distinguish between 5 cases depending on
  the type of transition relation (in \defref{def:ReminderAut}) that
  is traversed for the last time to reach $\rs$ from $(\macaxiom,
  \emptyset)$.

  Case \eqref{it:eext}: $\var \rightsquigarrow \word_{0} \var_{1}
  \word_{1} \cdots \var_{\prlen} \word_{\prlen}$ is a production with
  $\prlen \ge 1$ and $\rs_{1}\cdot (\var, \emptyset) \xLongrightarrow{
  \word_{0} \word_{1} \cdots \word_{\prlen}} \rs_{1} \cdot
  (\var_{\idxh}, \Lbrack \var_{1}, \dots, \var_{\prlen} \Rbrack
  \ominus \Lbrack \var_{\idxh} \Rbrack) = \rs$. Claim
  \ref{it:prdnPresent} is satisfied since $\Lbrack \var_{1}, \dots,
  \var_{\prlen} \Rbrack \ominus \Lbrack \var_{\idxh} \Rbrack \ne
  \emptyset$ implies that $\prlen \ge 2$. Claim \ref{it:noSucc} is
  satisfied since by induction hypothesis, none of the reminder pairs
  in $\rs_{1}$ can be of the form $(\var',\emptyset)$ for any $\var' \in
  \vars$. Claim \ref{it:descReach} is satisfied since $\var
  \rightarrow \var'$ for any $\var' \in \{\var_{1}, \dots,
  \var_{\prlen}\}$.

  Case \eqref{it:neext}: $\varmults \ne \emptyset$, $\var
  \rightsquigarrow \word_{0} \var_{1} \word_{1} \cdots \var_{\prlen}
  \word_{\prlen}$ is a production with $\prlen \ge 1$ and
  $\rs_{1}\cdot (\var, \varmults) \xLongrightarrow{ \word_{0} \word_{1}
  \cdots \word_{\prlen}} \rs_{1} \cdot ( \var, \varmults) \cdot
  (\var_{\idxh}, \Lbrack \var_{1}, \dots, \var_{\prlen} \Rbrack
  \ominus \Lbrack \var_{\idxh} \Rbrack) = \rs$. Claim
  \ref{it:prdnPresent} is satisfied since $\Lbrack \var_{1}, \dots,
  \var_{\prlen} \Rbrack \ominus \Lbrack \var_{\idxh} \Rbrack \ne
  \emptyset$ implies that $\prlen \ge 2$. Claim \ref{it:noSucc} is
  satisfied since by induction hypothesis, none of the reminder pairs
  in $\rs_{1}$ can be of the form $(\var',\emptyset)$ for any $\var' \in
  \vars$. Claim \ref{it:descReach} is satisfied since $\var
  \rightarrow \var'$ for any $\var' \in \{\var_{1}, \dots,
  \var_{\prlen}\}$.

  Case \eqref{it:erepl}: $\var \rightsquigarrow \word$ is a
  production and $\rs_{1} \cdot (\var, \emptyset)
  \xLongrightarrow{\word} \rs_{1} \cdot (\bot, \emptyset) = \rs$.
  Claim \ref{it:prdnPresent} is satisfied by any reminder pair in
  $\rs_{1}$ by induction hypothesis and is vacuously true for
  $(\bot, \emptyset)$. Claim \ref{it:noSucc} is satisfied since by
  induction hypothesis, none of the reminder pairs in $\rs_{1}$ can be
  of the form $(\var',\emptyset)$ for any $\var' \in \vars$. Claim
  \ref{it:descReach} is satisfied since it is satisfied in
  $\rs_{1}$ by induction hypothesis and $\bot$ is not a variable.

  Case \eqref{it:nerepl}: $\var \rightsquigarrow \word$ is a
  production and $\rs_{1} \cdot (\var, \varmults \oplus \Lbrack \var'
  \Rbrack) \xLongrightarrow{\word} \rs_{1} \cdot (\var', \varmults) =
  \rs$. By induction hypothesis, there is a production $\var''
  \rightsquigarrow \word_{0} \var_{1} \word_{1} \cdots \var_{\prlen}
  \word_{\prlen}$ with $\prlen \ge 2$ such that $\{\var, \var'\} \cup
  \{\var''' \mid \varmults(\var''') \ge 1\} \subseteq \{\var_{1}, \dots,
  \var_{\prlen}\}$. Hence, $\rs$ satisfies claim \ref{it:prdnPresent}.
  Claim \ref{it:noSucc} is satisfied since by induction hypothesis,
  none of the reminder pairs in $\rs_{1}$ can be of the form
  $(\var''',\emptyset)$ for any $\var''' \in \vars$. Claim
  \ref{it:descReach} is satisfied by induction hypothesis.

  Case \eqref{it:shorten}: $\rs_{1} \cdot (\var, \varmults \oplus \Lbrack
  \var' \Rbrack) \cdot (\bot, \emptyset) \xLongrightarrow{\epsilon}
  \rs_{1} \cdot ( \var', \varmults) = \rs$. By induction hypothesis,
  there is a production $\var'' \rightsquigarrow \word_{0} \var_{1}
  \word_{1} \cdots \var_{\prlen} \word_{\prlen}$ with $\prlen \ge 2$
  such that $\{\var, \var'\} \cup \{\var''' \mid \varmults(\var''') \ge
  1\} \subseteq \{\var_{1}, \dots, \var_{\prlen}\}$. Hence, $\rs$
  satisfies claim \ref{it:prdnPresent}.  Claim \ref{it:noSucc} is
  satisfied since by induction hypothesis, none of the reminder pairs
  in $\rs_{1}$ can be of the form $(\var''',\emptyset)$ for any
  $\var''' \in \vars$. Claim \ref{it:descReach} is satisfied by
  induction hypothesis. This completes the induction step and hence
  the claims \ref{it:prdnPresent}, \ref{it:noSucc},
  \ref{it:descReach} and \ref{it:numMembFolup} are true.

  Now we are ready to prove the lemma. Let $\rs$ be any state. By
  claim \ref{it:prdnPresent}, if $|\rs\{\idxo\}| > 1$, then there is
  some production $\var \rightsquigarrow \word_{0} \var_{1} \word_{1}
  \cdots \var_{\prlen} \word_{\prlen}$ such that $\rs\{\idxo\}
  \subseteq \{\var_{1}, \dots, \var_{\prlen}\}$ and $\prlen \ge 2$. By
  \eqref{eq:siblingEdge} in \defref{def:reminderGraph}, the reminder
  graph has an edge between every pair of variables in $\rs\{\idxo\}$.
  Let $1\le \idxo < \idxt \le |\rs|$. It only remains to prove that
  the reminder graph has an edge between any variable in
  $\rs\{\idxo\}$ and any variable in $\rs\{\idxt\}$. By claim
  \ref{it:noSucc}, $\rp_{\idxo}.\folup \ne \emptyset$ and by claim
  \ref{it:prdnPresent}, there is a production $\var'' \rightsquigarrow
  \word_{0} \var_{1} \word_{1} \cdots \var_{\prlen} \word_{\prlen}$
  with $\prlen \ge 2$ such that $\rs\{\idxo\} \subseteq \{\var_{1},
  \dots, \var_{\prlen}\}$. From claim \ref{it:descReach}, for any
  variable $\var \in \rs\{\idxt\}$, $\rp_{\idxo}.\cur \reachrel \var$.
  Since $\rp_{\idxo}.\cur \in \rs\{\idxo\}$, \eqref{eq:descendandEdge}
  in \defref{def:reminderGraph} implies that the reminder graph has an
  edge between any variable in $\rs\{\idxo\}$ and any variable in
  $\rs\{\idxt\}$.
\end{proof}

If $\rs$ is a state of $\aut( \cfg)$, then any tree decomposition of
the reminder graph of $\cfg$ will have a bag containing all variables
occurring in $\rs$, by \lemref{lem:autStateRemGraphClique} and
\lemref{lem:cliqueTreeDecomp}.  Since no variable can occur more than
twice in $\rs$, this gives us the required upper bound on the size of
$\aut( \cfg)$.
\begin{lemma}
  \label{lem:autSize}
  For a \ac{CFG} $\cfg$ with regularity width $\regms$, degree
  $\degree$ and $\numvars$ variables, the number of states of $\aut (
  \cfg)$ is bounded by a polynomial of $\numvars(\regms^{2\regms
  (\degree + 1)}) \maxtermins |\prodns|$, where $\maxtermins$ is the
  maximum number of terminal occurrences in the right hand side of any
  production and $|\prodns|$ is the number of productions.
\end{lemma}
\begin{proof}
  By \lemref{lem:autStateRemGraphClique}, the set of all variables
  $\vars_{1}$ occurring in a reminder sequence $\rs$ that is a state
  of $\aut ( \cfg)$ induce a clique in the reminder graph of $\cfg$.
  By \lemref{lem:cliqueTreeDecomp}, any tree decomposition (and hence
  an optimal tree decomposition) of the reminder graph of $\cfg$ has a
  bag that contains all variables in $\vars_{1}$, so $|\vars_{1}| \le
  \regms$.  Since any one variable can occur at most twice in
  $\rs$, there are at most $1 \regms$ reminder pairs in $\rs$.  For
  any graph, there is a optimal tree decomposition in which the number
  of nodes is at most the number of vertices in the graph \cite[Lemma
  11.9]{FG2006}. Consider such a tree decomposition of the reminder
  graph of $\cfg$ where each bag has at most $\regms$ variables and
  there are at most $\numvars$ bags.  The number of reminder sequences
  of lengh at most $2 \regms$ that can be constructed from variables
  in any one bag of this tree decomposition is at most
  $\regms^{2\regms (\degree + 1)}$. Since there are at most $\numvars$
  nodes in this tree decomposition, the number of states in $\aut (
  \cfg)$ is upper bounded by $\numvars (\regms^{2\regms (\degree +
  1)})$.

  The automaton $\aut ( \cfg)$ has transition relations that read
  words. If we need the usual automaton that reads single letters,
  the size will increase by a multiplicative factor that is a
  polynomial in $|\prodns|$ and $\maxtermins$.
\end{proof}

The idea behind $\aut ( \cfg)$ is that it will go down
one path of a parse tree, remembering siblings along the path that
will have to be visited later. In addition, if the automaton goes down
a parse tree $\prstr$ rooted at a variable $\var$,
\lemref{lem:parseTreeReduceRecurrence} is used to make one of the
subtrees of $\prstr$ $\var$-recurrence free so that $\aut ( \cfg)$ can
go down this subtree without having to remember any more siblings
rooted at $\var$. However, the application of
\lemref{lem:parseTreeReduceRecurrence} may introduce recurrences in
other variables. To handle this, we need the notion of a valuation.
The set of multisets over the set of parse trees is denoted by
$\prstrmultss$.
\begin{definition}
  \label{def:autStVal}
  Let $\rs$ be a reminder sequence. A \textbf{valuation} $\val$ for
  $\rs$ is a mapping $\val: \{1, \dots, |\rs|\} \to \prstrmultss$ such
  that for each $1\le \idxo \le |\rs|$, $\val ( \idxo) = \Lbrack
  \prstr_{1}, \dots, \prstr_{\prlen} \Rbrack$ implies
  $\rp_{\idxo}.\folup = \Lbrack \troot( \prstr_{1}), \dots, \troot(
  \prstr_{\prlen}) \Rbrack$. The Parikh image of such a valuation is
  defined as $\prkimg ( \val) = \sum_{1 \le \idxo \le |\rs|,
  \val(\idxo) ( \prstr) \ge 1}\prkimg(\yield ( \prstr) ^{ \val (
  \idxo) ( \prstr)})$ the sum of Parikh images of yields of all
  parse trees occurring in $\val$. A parse tree $\prstr$ is said to
  occur in $\val$ at level $\idxo$ if $\val ( \idxo) (
  \prstr) \ge 1$.
\end{definition}

Analogous to the notion of compact parse trees introduced in
\cite{EGKL2010}, we define the notion of compact valuations.
\begin{definition}
  \label{def:compVal}
  A valuation $\val$ for a reminder sequence $\rs$ is defined to be
  \textbf{compact} if the following properties are satisfied:
  \begin{enumerate}[{CP.}I]
    \item \label{it:cpOccFree} If $\var \in \vars$ is such that
      $\rp_{\idxt}.\cur = \rp_{\idxh}.\cur = \var$, $1 \le \idxt <
      \idxh \le |\rs|$ and $\rp_{\idxh}.\folup \ne \emptyset$, then
      all parse trees occurring in $\val$ at level $\idxh + 1$ or
      higher are $\var$-occurrence free.
    \item \label{it:cpSingleOccFree} For every $1\le \idxo < |\rs|$,
      if all parse trees occurring in $\val$ at level $\idxo + 1$ or
      lower are $\rp_{\idxo}.\cur$-occurrence free, then all parse
      trees occurring in $\val$ at level $\idxo + 2$ or higher are
      $\rp_{\idxo}.\cur$-occurrence free.
    \item \label{it:cpRecFree} For every $1\le \idxo < |\rs|$, if
      there is a parse tree occurring in $\val$ at level $\idxo + 1$
      or lower that is not $\rp_{\idxo}.\cur$-occurrence free, then
      all parse trees occurring in $\val$ at level $\idxo + 2$ or
      higher are $\rp_{\idxo}.\cur$-recurrence free.
  \end{enumerate}
\end{definition}
Now we will try to give some intuition behind the above definition.
Suppose $\aut (\cfg)$ at state $\rs$ is at the root of a parse tree
$\prstr$ whose children are roots of parse trees $\prstr_{1}, \dots,
\prstr_{\prlen}$. Then $\aut ( \cfg)$ will go down one of the
subtrees, say $\prstr_{1}$. This fact is remembered by the last
reminder pair $\rp_{|\rs|}$ by setting $\rp_{|\rs|}.\cur = \troot (
\prstr_{1})$ and $\rp_{|\rs|}.\folup = \Lbrack \troot ( \prstr_{2}),
\dots, \troot(\prstr_{\prlen}) \Rbrack$, which intuitively means that
$\troot ( \prstr_{1})$ is the label of the root of the subtree that is
currently being handled, and subtrees rooted at $\troot ( \prstr_{2}),
\dots, \troot ( \prstr_{\prlen})$ are to be followed up later. The
property CP.\ref{it:cpOccFree} above means that if $\aut ( \cfg)$ has
already seen a variable $\var$ twice in the subtree currently being
handled, first at level $\idxt$ and then at level $\idxh$, then $\var$
will never be seen again in the current subtree (subtrees of the
current subtree will end up in level $\idxh + 1$ or higher). This will
ensure that $\aut ( \cfg)$ will not need reminder sequences that are
too long. To ensure that CP.\ref{it:cpOccFree} is always maintained,
$\aut ( \cfg)$ has to be careful about which subtree to go into at
each stage. This is captured by CP.\ref{it:cpRecFree}: if at stage
$\idxo$, $\aut ( \cfg)$ is at a subtree rooted at $\rp_{\idxo}.\cur$,
then all further recurrences of $\rp_{\idxo}.\cur$ are moved (using
\lemref{lem:parseTreeReduceRecurrence}) into one of the children of
the current subtree and moved into $\rp_{\idxo + 1}.\folup$ to be
followed up later, ensuring that the subtree being handled right now
are free of $\rp_{\idxo}.\cur$-recurrences (and hence any subtree that
gets into $\idxo + 2$ or higher levels are free of
$\rp_{\idxo}.\cur$-recurrences too). The property
CP.\ref{it:cpSingleOccFree} captures the fact that if at stage
$\idxo$, \lemref{lem:parseTreeReduceRecurrence} was not used to push
all $\rp_{\idxo}.\cur$-recurrences into one of the subtrees at $\idxo
+ 1$\textsuperscript{th} stage (so that all parse trees occurring at
level $\idxo + 1$ or lower are $\rp_{\idxo}.\cur$-occurrence free), it
was because none of the subtrees had any occurrence of
$\rp_{\idxo}.\cur$ at all (so that all parse trees occurring at level
$\idxo + 2$ or higher are also $\rp_{\idxo}.\cur$-occurrence free).

If $\rs$ is a state of $\cfg$, $\val$ is a valuation for $\rs$ and $1
\le \idxo \le |\rs|$, then $\rs\restr \idxo$ is the reminder sequence
$\rp_{1} \cdots \rp_{\idxo}$ and $\val \restr \idxo$ is the
restriction of $\val$ to $\{1, \dots, \idxo\}$. We denote by $\vars[
\val \uparrow \idxo] = \{\var \in \vars \mid \exists \idxt \ge \idxo,
\var \text{ occurs in parse tree } \prstr, \val( \idxt) ( \prstr) \ge
1\}$ the set of all variables labelling parse trees that occur in
$\val$ at level $\idxo$ or higher. Similarly, $\vars[ \val \downarrow
\idxo] = \{\var \in \vars \mid \exists \idxt \le \idxo, \var \text{
occurs in parse tree } \prstr, \val( \idxt) ( \prstr) \ge 1\}$ is the
set of all variables labelling parse trees that occur in $\val$ at
level $\idxo$ or lower.
\begin{proposition}
  \label{prop:cpPreserve}
  If $\val$ is a compact valuation for $\rs$, then $\val \restr
  ( |\rs| - 1)$ is a compact valuation for $\rs \restr (|\rs| - 1)$.
\end{proposition}
\begin{proof}
  If $\var \in \vars$ is such that $\rp_{\idxt}.\cur =
  \rp_{\idxh}.\cur = \var$, $1 \le \idxt < \idxh \le |\rs|$,
  $\rp_{\idxh}.\folup \ne \emptyset$ and $\idxh \le |\rs| - 1$, then
  by compactness of $\val$, all parse trees occurring in $\val \restr
  ( |\rs| - 1)$ at level $\idxh + 1$ or higher are $\var$-occurrence
  free. Therefore, $\val \restr ( |\rs| - 1)$ satisfies
  CP.\ref{it:cpOccFree}.

  For every $1\le \idxo < |\rs|-1$, if all parse trees occurring in
  $\val$ at level $\idxo + 1$ or lower are
  $\rp_{\idxo}.\cur$-occurrence free, then by compactness of $\val$,
  all parse trees occurring in $\val \restr (|\rs| - 1)$ at level
  $\idxo + 2$ or higher are $\rp_{\idxo}.\cur$-occurrence free.  Hence
  $\val \restr ( |\rs| - 1)$ satisfies CP.\ref{it:cpSingleOccFree}.

  For every $1\le \idxo < |\rs|-1$, if there is a parse tree occurring
  in $\val$ at level $\idxo + 1$ or lower that is not
  $\rp_{\idxo}.\cur$-occurrence free, then by compactness of
  $\val$, all parse trees occurring in $\val \restr ( |\rs| - 1)$ at
  level $\idxo + 2$ or higher are $\rp_{\idxo}.\cur$-recurrence free.
  Hence $\val \restr ( |\rs| - 1)$ satisfies CP.\ref{it:cpRecFree}.
\end{proof}
\begin{lemma}
  \label{lem:valCompactification}
  If a valuation $\val$ for a reminder sequence $\rs$ satisfies
  properties CP.\ref{it:cpOccFree} and CP.\ref{it:singleOccFree} of
  \defref{def:compVal}, then there is a compact valuation $\val'$ for
  $\rs$ such that $\prkimg ( \val') = \prkimg ( \val)$ and for all
  $1\le \idxo \le |\rs|$, $\vars[ \val \downarrow \idxo] \subseteq
  \vars[ \val' \downarrow \idxo]$.
\end{lemma}
\begin{proof}
  By induction on $|\rs|$. For the base case $|\rs| =1$, the property
  CP.\ref{it:cpRecFree} is vacuously true for $\val$ and hence it is
  compact. 

  For the induction step, we will describe a compactification
  procedure.
  
  Step 1: Suppose for some $1 \le \idxo < |\rs| - 1$, there is a parse
  tree $\prstr_{1}$ occurring in $\val$ at level $\idxo ' \le \idxo +
  1$ that is not $\rp_{\idxo}.\cur$-occurrence free and there is a
  parse tree $\prstr_{2}$ occurring in $\val$ at level $|\rs|$ that is
  not $\rp_{\idxo}.\cur$-recurrence free (we will call such a pair of
  parse trees $(\prstr_{1}, \prstr_{2})$ a \emph{CP.\ref{it:cpRecFree}
  violation witness}). Since there is a path from the root to a leaf
  of $\prstr_{2}$ in which $\rp_{\idxo}.\cur$ occurs twice, we have
  $\prstr_{2} = \prstr_{2}' \cdot \prstr' \cdot \prstr_{2}''$ where
  $\prstr'$ and $\prstr_{2}''$ are rooted at $\rp_{\idxo}.\cur$.
  Similarly, $\prstr_{1} = \prstr_{1}' \cdot \prstr_{1}''$, where
  $\prstr_{1}''$ is rooted at $\rp_{\idxo}.\cur$. Replace $\prstr_{1}$
  by $\prstr_{1}' \cdot \prstr' \cdot \prstr_{1}''$ and $\prstr_{2}$
  by $\prstr_{2}' \cdot \prstr_{2}''$ and let the resulting valuation
  be $\val_{1}$. The insertion of $\prstr'$ into $\prstr_{1}$ will not
  introduce any violation of CP.\ref{it:cpOccFree}: if there are $1
  \le \idxt < \idxh < \idxo'$ with $\rp_{\idxh}.\folup \ne \emptyset$
  and $\rp_{\idxt}.\cur = \rp_{\idxh}.\cur = \var \in \vars$, then
  $\prstr_{2}$ is $\var$-occurrence free and hence so is $\prstr'$.
  The insertion of $\prstr'$ into $\prstr_{1}$ will not introduce any
  violation of CP.\ref{it:cpSingleOccFree}: for any $1 \le \idxt <
  \idxo' - 1$, if all parse trees occurring at level $\idxt + 1$ or
  lower are $\rp_{\idxt}.\cur$-occurrence free, then $\prstr_{2}$ is
  $\rp_{\idxt}.\cur$-occurrence free and hence so is $\prstr'$.  The
  removal of $\prstr'$ from $\prstr_{2}$ also does not introduce any
  violation of CP.\ref{it:cpOccFree} or CP.\ref{it:cpSingleOccFree}
  since removal of subtrees does not introduce new labels. Hence,
  $\val_{1}$ continues to satisfy CP.\ref{it:cpOccFree} and
  CP.\ref{it:cpSingleOccFree}. In addition, $\prkimg( \val_{1}) =
  \prkimg ( \val)$ and for all $1\le \idxt \le |\rs|$, $\vars[ \val
  \downarrow \idxt] \subseteq \vars[ \val_{1} \downarrow \idxt]$.

  Step 2: Iterate step 1 as long as there are CP.\ref{it:cpRecFree}
  violation witnesses. Since each iteration reduces the number of
  nodes in one of the parse trees at level $|\rs|$, this loop will
  stop after a finite number of rounds.  At this stage, suppose we
  have the updated valuation $\val_{\prlen}$ that does not have any
  CP.\ref{it:cpRecFree} violation witnesses and satisfies
  CP.\ref{it:cpOccFree} and CP.\ref{it:cpSingleOccFree}. Also suppose
  that $\Lbrack \prstr_{1}', \dots, \prstr_{\numtr}' \Rbrack$ are the
  trees remaining behind at level $|\rs|$ after removal of subtrees
  from the original ones. The valuation $\val_{\prlen} \restr (
  |\rs|-1)$ for the reminder sequence $\rs \restr ( |\rs| - 1)$
  satisfies CP.\ref{it:cpOccFree} and CP.\ref{it:cpSingleOccFree}
  (refer to proof of \propref{prop:cpPreserve}) and hence by induction
  hypothesis, there is a compact valuation $\val_{\prlen}'$ for $\rs
  \restr ( |\rs| - 1)$ such that $\prkimg ( \val_{\prlen}') = \prkimg
  ( \val_{\prlen} \restr ( |\rs| - 1))$ and for all $1\le \idxt \le
  |\rs|-1$, $\vars[ \val_{\prlen} \restr (|\rs| - 1) \downarrow \idxt]
  \subseteq \vars[ \val_{\prlen}' \downarrow \idxt]$. Let $\val'$ be
  the valuation for $\rs$ such that for each $1 \le \idxo < |\rs|$,
  $\val' ( \idxo ) = \val_{\prlen}' ( \idxo)$ and $\val' ( |\rs|) =
  \Lbrack \prstr_{1}', \dots, \prstr_{\numtr}' \Rbrack$. It is clear
  that $\prkimg ( \val') = \prkimg ( \val)$ and for all $1\le \idxt
  \le |\rs|$, $\vars[ \val \downarrow \idxt] \subseteq \vars[ \val'
  \downarrow \idxt]$. We will now prove that $\val'$ satisfies
  CP.\ref{it:cpOccFree}. If not, there is $\var \in \vars$ such that
  $\rp_{\idxt}.\cur = \rp_{\idxh}.\cur = \var$, $\rp_{\idxh}.\folup =
  \emptyset$ and there is a parst tree $\prstr$ occurring in $\val'$
  at level $\idxh + 1$ or higher that is not $\var$-occurrence free.
  Due to the compactness of $\val_{\prlen}'$, $\prstr$ can not occur
  at level $|\rs| - 1$ or lower, hence $\prstr$ occurs in $\val'$ at
  level $|\rs|$. Then $\prstr$ is one among $\Lbrack \prstr_{1}',
  \dots, \prstr_{\numtr}' \Rbrack$, say $\prstr_{1}'$, which is
  obtained from $\prstr_{1}$ by removing subtrees, where $\prstr_{1}$
  occurs in $\val$ at level $|\rs|$.  Hence, $\prstr_{1}$ is not
  $\var$-occurrence free, contradicting the fact that $\val$ satisfies
  CP.\ref{it:cpOccFree}. Hence, $\val'$ satisfies
  CP.\ref{it:cpOccFree}. We will show that $\val'$ satisfies
  CP.\ref{it:cpSingleOccFree}. If not, for some $1 \le \idxo <|\rs|$,
  all parse trees occurring in $\val'$ at level $\idxo + 1$ or lower
  are $\rp_{\idxo}.\cur$-occurrence free and there is a parse tree
  $\prstr$ occurring at level $\idxo + 2$ or higher that is not
  $\rp_{\idxo}.\cur$-occurrence free. Due to the compactness of
  $\val_{\prlen}'$, $\prstr$ can not occur at level $|\rs| - 1$ or
  lower, hence $\prstr$ occurs in $\val'$ at level $|\rs|$. Then
  $\prstr$ is one among $\Lbrack \prstr_{1}', \dots, \prstr_{\numtr}'
  \Rbrack$, say $\prstr_{1}'$, which is obtained from $\prstr_{1}$ by
  removing subtrees, where $\prstr_{1}$ occurs in $\val$ at level
  $|\rs|$.  Hence, $\prstr_{1}$ is not $\var$-occurrence free,
  contradicting the fact that $\val$ satisfies CP.\ref{it:cpOccFree}
  (recall that $\vars[ \val \downarrow \idxo + 1] \subseteq
  \vars[ \val' \downarrow \idxo + 1]$).
  Hence, $\val'$ satisfies CP.\ref{it:cpSingleOccFree}.

  Step 3: The valuation $\val_{\prlen}$ obtained in step 2 does not
  have CP.\ref{it:cpRecFree} violation witnesses. But $\val'$ is
  obtained from $\val_{\prlen}$ by performing some changes, so $\val'$
  may have CP.\ref{it:cpRecFree} violation witnesses. Iterate step 2
  as long as the resulting $\val'$ has CP.\ref{it:cpRecFree} violation
  witnesses. Since each iteration reduces the number of nodes in parse
  trees occurring at level $|\rs|$, this loop will stop after a finite
  number of rounds. The resulting valuation $\val'$ has no
  CP.\ref{it:cpRecFree} violation witnesses, satisfies
  CP.\ref{it:cpOccFree} and CP.\ref{it:cpSingleOccFree}, $\prkimg (
  \val') = \prkimg ( \val)$ and $\val' \restr (|\rs| - 1)$ is compact.
  In addition, for all $1\le \idxt \le |\rs|$, $\vars[ \val \downarrow
  \idxt] \subseteq \vars[ \val' \downarrow \idxt]$. We will prove that
  $\val'$ satisfies CP.\ref{it:cpRecFree} (and hence compact). Suppose
  not. For some $1 \le \idxo < |\rs|$, there is a parse tree
  $\prstr_{1}$ occurring in $\val'$ at level $\idxo + 1$ or lower that
  is not $\rp_{\idxo}.\cur$-occurrence free and there is a parse tree
  $\prstr_{2}$ occurring at level $\idxo + 2$ or higher that is not
  $\rp_{\idxo}.\cur$-recurrence free. Since $\val' \restr ( |\rs| -
  1)$ is compact, $\prstr_{2}$ can not occur at level $|\rs| - 1$ or
  lower, so it has to occur at level $|\rs|$. But then, $(\prstr_{1},
  \prstr_{2})$ is a CP.\ref{it:cpRecFree} violation witness, a
  contradiction. Hence, $\val'$ satisfies CP.\ref{it:cpRecFree} and
  hence it is compact.
\end{proof}

The automaton $\aut ( \cfg)$ deletes some portions of parse trees and
accepts a sequence of terminals corresponding to the deleted portion.
To track the progress of $\aut ( \cfg)$, we define below a function
$\info$ that measures the amount of information contained in a parse
tree. For convenience, we consider $\bot$ as a special parse tree such
that root is its only node and is rooted at $\bot$. The yield of
$\bot$ is defined to be $\yield( \bot) = \zv$.
\begin{definition}
  \label{def:info}
  The function $\info$ from the set of parse trees to the set of
  natural numbers is defined as follows:
  \begin{itemize}
    \item $\info ( \bot) = 1$.
    \item $\info (\prstr) = 2$ if the root of $\prstr$ is a variable
      and all its children are labelled by terminals. This case
      applies even if the root of $\prstr$ does not have any children.
    \item If a parse tree $\prstr$ has at least one child labelled by
      a variable, then $\info ( \prstr) = 1 + \sum_{1\le \idxo \le
      \prlen}\info ( \prstr_{\idxo})$, where $\prstr_{1}, \dots,
      \prstr_{\prlen}$ are the subtrees of $\prstr$ whose roots are
      children of the root of $\prstr$ and that are labelled by
      variables.
  \end{itemize}
  The function $\info$ is extended to the domain of valuations as
  follows. If $\val$ is a valuation for a reminder sequence $\rs$,
  $\info ( \val) = \sum_{1\le \idxo \le |\rs|,\val( \idxo) ( \prstr)
  \ge 1}\info ( \prstr)\val ( \idxo ) ( \prstr)$.
\end{definition}

To describe runs of our automaton, we maintain a valuation for its
current state apart from the word run over till the current state.
This is formalized in the definition below.
\begin{definition}
  \label{def:config}
  A \textbf{configuration} $\config$ is a tuple $(\rs, \word, \val,
  \prstr)$ where $\rs$ is a state of $\aut ( \cfg)$, $\word \in
  \termins^{*}$ is a word, $\val$ is a compact valuation for $\rs$ and
  $\prstr$ is a parse tree rooted at $\rp_{|\rs|}.\cur$. The size
  $|\config|$ of $\config$ is defined to be $\info ( \val) + \info
  ( \prstr)$.
  %the ordered triple
  %$(\prkimg(\val) + \prkimg ( \yield(\prstr)), |\rs|, \height(
  %\prstr))$. Let $\config_{1} = (\rs_{1}, \word_{1}, \val_{1},
  %\prstr_{1})$ and $\config_{2} = (\rs_{2}, \word_{2}, \val_{2},
  %\prstr_{2})$. We use the lexicographic order on the ordered triple
  %$(\prkimg(\val) + \prkimg ( \yield(\prstr)), |\rs|, \height ( \prstr))$ to
  %compare sizes of configurations, i.e., $|\config_{1}| <
  %|\config_{2}|$ if either $\prkimg( \val_{1}) + \prkimg ( \yield(
  %\prstr_{1})) < \prkimg( \val_{2}) + \prkimg ( \yield ( \prstr_{2}))$
  %or $\prkimg( \val_{1}) + \prkimg ( \yield( \prstr_{1})) = \prkimg(
  %\val_{2}) + \prkimg ( \yield ( \prstr_{2}))$ and $|\rs_{1}| <
  %|\rs_{2}|$ or $\prkimg( \val_{1}) + \prkimg ( \yield( \prstr_{1})) =
  %\prkimg( \val_{2}) + \prkimg ( \yield ( \prstr_{2}))$ and $|\rs_{1}|
  %= |\rs_{2}|$ and $\height( \prstr_{1}) < \height ( \prstr_{2})$.
\end{definition}
%Since $\prkimg( \val_{1}) + \prkimg( \yield(\prstr_{1})) < \prkimg( \val_{2})
%+ \prkimg ( \yield( \prstr_{2})) $ is a well quasi ordering on Parikh images
%(Dickson's lemma) and $|\rs_{1}| < |\rs_{2}|, \height (
%\prstr_{1}) < \height ( \prstr_{2})$ is the well ordering on
%natural numbers, there cannot be an infinite sequence of strictly
%decreasing configurations.

Now we will show that the transitions of $\aut ( \cfg)$ can be used to
traverse a parse tree in such a way that the size of reminder
sequences needed never exceed a small bound.  If $\rs$ is a reminder
sequence and $\var \in \vars$ is a variable, we say that $\rs$ ends
with $\var$ when $\rp_{|\rs|}.\cur = \var$. For a parse tree $\prstr$,
$\word ( \prstr)$ is the word over $\termins$ obtained by concatening
the labels of those children of $\troot ( \prstr)$ that are labelled
by terminals. If the root does not have any children or all children
of the root are labelled by variables, then $\word ( \prstr) =
\epsilon$. Immediate subtrees of $\prstr$ are those subtrees of
$\prstr$ whose roots are children of the root of $\prstr$.
%The proof
%of the following lemma has been moved to the appendix due to space
%constraints.
The proof of the following lemma is based on the intuition given after
\defref{def:compVal}.
\begin{lemma}
  \label{lem:autRun}
  Suppose $\cfg$ is a \ac{CFG} and $\config = (\rs, \word, \val,
  \prstr)$ is a configuration satisfying the following properties:
  \begin{enumerate}[I]
    \item \label{it:occFree} If $\var \in \vars$ is such that
      $\rp_{\idxt}.\cur = \rp_{\idxh}.\cur = \var$, $1\le \idxt <
      \idxh \le |\rs|$
      and $\rp_{\idxh}.\folup \ne \emptyset$, then all immediate
      subtrees of $\prstr$ are $\var$-occurrence free.
    \item \label{it:singleOccFree} For every $1\le \idxo < |\rs|$, if
      all parse trees occurring in $\val$ at level $\idxo + 1$ or
      lower are $\rp_{\idxo}.\cur$-occurrence free and $\rp_{\idxo +
      1}.\folup \ne \emptyset$, then $\prstr$ is
      $\rp_{\idxo}.\cur$-occurrence free.
    \item \label{it:recFree} For every $1\le \idxo < |\rs|$, if there
      is a parse tree occurring in $\val$ at level $\idxo + 1$ or
      lower that is not $\rp_{\idxo}.\cur$-occurrence free and
      $\rp_{\idxo + 1}.\folup \ne \emptyset$, then
      $\prstr$ is $\rp_{\idxo}.\cur$-recurrence free.
      \newcounter{enumiisaved}
      \setcounter{enumiisaved}{\value{enumi}}
  \end{enumerate}
  If the size of such a configuration $\config$ is greater than
  $1$, then there is another configuration $\config' = (\rs',
  \word \word', \val', \prstr')$ that satisfies the above three
  properties in addition to the following ones:
  \begin{enumerate}[I]
    \setcounter{enumi}{\value{enumiisaved}}
  \item \label{it:autTrans} $\rs \xLongrightarrow{ \word'} \rs'$,
  \item \label{it:sizeDec} $|\config'| < |\config|$ and
  \item \label{it:prkimgPresrv} $\prkimg(\val) + \prkimg( \word) +
    \prkimg ( \yield ( \prstr)) =
    \prkimg( \val') + \prkimg( \word \word') + \prkimg ( \yield
    (\prstr'))$.
  \end{enumerate}
\end{lemma}
\begin{proof}[Proof of \lemref{lem:autRun}]
  Let $\rs = \rs_{1} \cdot (\var, \varmults)$ with $\rs_{1}$ possibly
  equal to the empty sequence $\epsilon$ and $\var$ possibly equal to
  the special symbol $\bot$. We distinguish between cases based on
  whether $\varmults = \emptyset$ and whether some children of $\troot
  ( \prstr)$ are labelled with variables. In all the following cases,
  properties \ref{it:autTrans}, \ref{it:sizeDec} and
  \ref{it:prkimgPresrv} are clearly satisfied, so they will not be
  mentioned explicitly.

  Case 1: $\varmults = \emptyset$ and all children of the root of $
  \prstr$ are labelled with terminals. In this case, $\var
  \rightsquigarrow \word ( \prstr)$ is a production of $\cfg$. We can
  take $\rs' = \rs_{1} \cdot (\bot, \emptyset)$, $\word' = \word (
  \prstr)$, $\val' = \val$ and $\prstr' = \bot$. The new configuration
  $\config'$ satisfies properties \ref{it:occFree}, \ref{it:singleOccFree}
  and \ref{it:recFree} since $\prstr' = \bot$.

  Case 2: $\varmults = \emptyset$ and at least one child of the root
  of $\prstr$ is labelled with a variable. In this case, $\var
  \rightsquigarrow \word_{0} \var_{1} \word_{1} \cdots \var_{\prlen}
  \word_{\prlen}$ is a production of $\cfg$ such that $\word ( \prstr)
  = \word_{0} \cdots \word_{\prlen}$. Let $\prstr_{1}, \dots,
  \prstr_{\prlen}$ be the immediate subtrees of $\prstr$ rooted at
  $\var_{1}, \dots, \var_{\prlen}$ respectively. If $\prlen \ge 2$ and
  $\rs_{1}$ ends with $\var' \in \vars$ such that one of the trees
  $\prstr_{1}, \dots, \prstr_{\prlen}$ is $\var'$-occurrence free,
  that one
  should be chosen as $\prstr'$ for the next configuration $\config'$.
  If $\prlen \ge 2$ and $\rs_{1}$ ends with $\var' \in \vars$ such
  that none of the the trees among $\prstr_{1}, \dots,
  \prstr_{\prlen}$ are $\var'$-occurrence free, make one of them
  $\var'$-recurrence free by applying
  \lemref{lem:parseTreeReduceRecurrence} and choose that as $\prstr'$
  for the next configuration $\config'$. Assume \ac{wlog} that if $\prlen \ge
  2$ and $\rs_{1}$ ends with $\var'$, then $\prstr_{1}$ is
  $\var'$-occurrence free when possible and $\var'$-recurrence free
  otherwise. Let $\rs' = \rs_{1} \cdot (\var_{1}, \Lbrack \var_{2},
  \dots, \var_{\prlen} \Rbrack)$. If variable $\var_{1}$ occurs more
  than once in $\rs_{1}$, then $\config$ would violate property
  \ref{it:occFree} that $\prstr_{1}$ is $\var_{1}$-occurrence free.
  Hence $\var_{1}$ occurs at most once in $\rs_{1}$ and hence $\rs'$
  is a state of $\aut ( \cfg)$. Let $\val_{1}$ be such that $\val_{1}
  \restr ( |\rs|-1) = \val \restr ( |\rs| - 1)$ and $\val_{1} ( |\rs|)
  = \Lbrack \prstr_{2}, \dots, \prstr_{\prlen} \Rbrack$. This new
  valuation $\val_{1}$ satisfies CP.\ref{it:cpOccFree}: since
  $\val_{1} \restr ( |\rs| - 1)$ already satisfies
  CP.\ref{it:cpOccFree} (due to compactness of $\val$), it is enough
  to observe that for any $\var' \in \vars$ with $\rp_{\idxt}.\cur =
  \rp_{\idxh}.\cur = \var'$ and $1 \le \idxt < \idxh \le |\rs_{1}|$,
  we have that $\prstr_{1}, \dots,
  \prstr_{\prlen}$ are $\var'$-occurrence free since $\config$
  satisfies property \ref{it:occFree}. The new valuation
  $\val_{1}$ satisfies CP.\ref{it:cpSingleOccFree}: since
  $\val_{1} \restr ( |\rs| - 1)$ already satisfies
  CP.\ref{it:cpSingleOccFree} (due to compactness of $\val$), it is enough
  to observe that for any $\idxt \le |\rs_{1}| - 1$, if all parse trees
  occurring in $\val_{1}$ at level $\idxt + 1$ or lower are
  $\rp_{\idxt}.\cur$-occurrence free, then so are all parse trees
  occurring in $\val$ at level $\idxt + 1$ or lower, and hence,
  $\prstr_{1}, \dots, \prstr_{\prlen}$ are
  $\rp_{\idxt}.\cur$-occurrence free since $\config$ satisfies
  property \ref{it:singleOccFree}.
  Let $\val'$ be the compact
  valuation given by \lemref{lem:valCompactification} such that
  $\prkimg ( \val') = \prkimg ( \val_{1})$. We take $\rs' = \rs_{1}
  \cdot ( \var_{1}, \Lbrack \var_{2}, \dots, \var_{\prlen} \Rbrack)$,
  $\word' = \word ( \prstr)$, $\val'$ is the one obtained above by
  applying \lemref{lem:valCompactification} and $\prstr' =
  \prstr_{1}$.

  The new configuration $\config' = (\rs', \word \word', \val',
  \prstr')$ given above satisfies property \ref{it:occFree}: let
  $\var' \in \vars$ be such that $\rp_{\idxt}.\cur = \rp_{\idxh}.\cur
  = \var'$, $1 \le \idxt < \idxh \le |\rs'|$ and $\rp_{\idxh}.\folup
  \ne \emptyset$. If $\idxh \le |\rs_{1}|$, then $\prstr_{1}$ is
  $\var'$-occurrence free since $\config$ satisfies property
  \ref{it:occFree}. If $\idxh = |\rs'|$ (in which case $\var' =
  \var_{1}$), $\idxt < |\rs_{1}|$ and all parse trees occurring in
  $\val$ at level $\idxt + 1$ or lower are $\var'$-occurrence free,
  then since $\config$ satisfies property \ref{it:singleOccFree},
  $\prstr$ is $\var'$-occurrence free and so is $\prstr_{1}$. If
  $\idxh = |\rs'|$, $\idxt < |\rs_{1}|$ and a parse tree occurring in
  $\val$ at level $\idxt + 1$ or lower is not $\var'$-occurrence free,
  then since $\config$ satisfies property \ref{it:recFree}, $\prstr$
  is $\var'$-recurrence free and hence immediate subtrees of
  $\prstr_{1}$ are $\var'$-occurrence free (since $\prstr_{1}$ is
  rooted at $\var' = \var_{1}$). If $\idxh = |\rs'|$ and $\idxt =
  |\rs_{1}|$, $\rp_{\idxh}.\folup \ne \emptyset$ implies that we chose
  $\prstr_{1}$ to be $\rp_{\idxt}.\cur$-recurrence free at the
  beginning of this case, hence immedate subtrees of $\prstr_{1}$ are
  $\rp_{\idxt}.\cur$-occurrence free. The new configuration also
  satisfies property \ref{it:singleOccFree}: suppose for any $1 \le
  \idxo < |\rs'| $, all parse trees occurring in $\val'$ at level
  $\idxo + 1$ or lower are $\rp_{\idxo}.\cur$-occurrence free and
  $\rp_{\idxo + 1}.\cur \ne \emptyset$.  If $\idxo < |\rs'| - 1$, then
  $\vars[ \val \downarrow \idxo + 1] = \vars [ \val_{1} \downarrow
  \idxo + 1] \subseteq \vars [ \val' \downarrow \idxo + 1]$ implies
  that all parse trees occurring in $\val$ at level $\idxo + 1$ or
  lower are $\rp_{\idxo}.\cur$-occurrence free.  Since $\config$
  satisfies property \ref{it:singleOccFree}, $\prstr$ is
  $\rp_{\idxo}.\cur$-occurrence free and so is $\prstr_{1}$. If $\idxo
  = |\rs'| - 1$, $\prlen \ge 2$ implies that $\prstr_{1}$ is
  $\rp_{\idxo}.\cur$-occurrence free by the choice of $\prstr_{1}$
  made at the beginning of this case. If $\idxo = |\rs'| - 1$, $\prlen
  = 1$ implies that $\rp_{\idxo + 1}.\folup = \emptyset$ so this case
  is not applicable.  Finally, the new configuration satisfies
  property \ref{it:recFree}: for some $1 \le \idxo < |\rs'|$, suppose
  there is a parse tree occurring in $\val'$ at level $\idxo + 1$ or
  lower that is not $\rp_{\idxo}.\cur$-occurrence free and
  $\rp_{\idxo + 1}.\folup \ne \emptyset$. If $\idxo <
  |\rs_{1}|$ and a parse tree occurring in $\val$ at level $\idxo + 1$
  or lower is not $\rp_{\idxo}.\cur$-occurrence free, then since
  $\config$ satisfies property \ref{it:recFree}, $\prstr$ is
  $\rp_{\idxo}.\cur$-recurrence free and hence so is $\prstr_{1}$. If
  $\idxo < |\rs_{1}|$ and all parse trees occurring in $\val$ at level
  $\idxo + 1$ or lower are $\rp_{\idxo}.\cur$-occurrence free, then
  since $\config$ satisfies property \ref{it:singleOccFree}, $\prstr$
  is $\rp_{\idxo}.\cur$-occurrence free and so is $\prstr_{1}$. If
  $\idxo = |\rs_{1}|$ and $\prlen = 1$, then $\rp_{\idxo + 1}.\folup =
  \emptyset$ so this case does not apply. If $\idxo =
  |\rs_{1}|$ and $\prlen \ge 2$, then $\prstr_{1}$ is
  $\rp_{\idxo}.\cur$-recurrence free by the choice of $\prstr_{1}$ we
  made at the beginning of this case.

  Case 3: $\rs = \rs_{1} \cdot ( \var, \Lbrack \var_{1}, \dots,
  \var_{\prlen}\Rbrack) \cdot (\bot, \emptyset)$. Let $\val (
  |\rs_{1}| + 1) = \Lbrack \prstr_{1}, \dots, \prstr_{\prlen}\Rbrack$
  such that $\prstr_{1}, \dots, \prstr_{\prlen}$ are rooted at
  $\var_{1}, \dots, \var_{\prlen}$ respectively. If $\prlen \ge 2$ and
  $\rs_{1}$ ends with $\var' \in \vars$ such that one of the trees
  $\prstr_{1}, \dots, \prstr_{\prlen}$ is $\var'$-occurrence free,
  that one
  should be chosen as $\prstr'$ for the next configuration $\config'$.
  If $\prlen \ge 2$ and $\rs_{1}$ ends with $\var' \in \vars$ such
  that none of the the trees among $\prstr_{1}, \dots,
  \prstr_{\prlen}$ are $\var'$-occurrence free, make one of them
  $\var'$-recurrence free by applying
  \lemref{lem:parseTreeReduceRecurrence} and choose that as $\prstr'$
  for the next configuration $\config'$. Assume \ac{wlog} that if $\prlen \ge
  2$ and $\rs_{1}$ ends with $\var'$, then $\prstr_{1}$ is
  $\var'$-occurrence free when possible and $\var'$-recurrence free
  otherwise.
  Let $\rs' =
  \rs_{1} \cdot (\var_{1}, \Lbrack \var_{2}, \dots, \var_{\prlen}
  \Rbrack)$. If variable $\var_{1}$ occurs more than once in
  $\rs_{1}$, it leads to a contradiction since compactness of $\val$
  implies that $\prstr_{1}$ is $\var_{1}$-occurrence free.  Hence
  $\var_{1}$ occurs at most once in $\rs_{1}$ and hence $\rs'$ is a
  state of $\aut ( \cfg)$. Let $\val_{1}$ be such that $\val_{1}
  \restr ( |\rs'|-1) = \val \restr ( |\rs'| - 1)$ and $\val_{1} (
  |\rs'|) = \Lbrack \prstr_{2}, \dots, \prstr_{\prlen} \Rbrack$. This
  new valuation $\val_{1}$ satisfies CP.\ref{it:cpOccFree} and
  CP.\ref{it:cpSingleOccFree} since it is
  obtained from $\val$ by moving subtrees within parse trees occurring
  at level $|\rs'|$ and removing a parse tree from the same level. Let
  $\val'$ be the compact valuation given by
  \lemref{lem:valCompactification} such that $\prkimg ( \val') =
  \prkimg ( \val_{1})$. We take $\rs' = \rs_{1} \cdot ( \var_{1},
  \Lbrack \var_{2}, \dots, \var_{\prlen} \Rbrack)$, $\word' =
  \epsilon$, $\val'$ is the one obtained above by applying
  \lemref{lem:valCompactification} and $\prstr' = \prstr_{1}$.

  The new configuration $\config' = (\rs', \word \word', \val',
  \prstr')$ given above satisfies property \ref{it:occFree}: let
  $\var' \in \vars$ be such that $\rp_{\idxt}.\cur = \rp_{\idxh}.\cur
  = \var'$, $1\le \idxt < \idxh \le |\rs'|$ and $\rp_{\idxh}.\folup
  \ne \emptyset$. If
  $\idxh < |\rs'|$, then $\prstr_{1}$ is $\var'$-occurrence free by
  compactness of $\val$. If $\idxh = |\rs'|$ (in which case $\var' =
  \var_{1}$) and $\idxt = |\rs_{1}|$, then $\rp_{\idxh}.\folup
  \ne \emptyset$ implies that $\prstr_{1}$ is $\var'$-recurrence free
  by the choice of $\prstr_{1}$ we made in the biginning of this case,
  hence immediate subtrees of $\prstr_{1}$ are $\var'$-occurrence
  free. If $\idxh = |\rs'|$, $\idxt < |\rs_{1}|$ and all parse trees
  occurring in $\val$ at level $\idxt + 1$ or lower are
  $\var'$-occurrence free, then $\prstr_{1}$ is $\var'$-occurrence
  free by compactness of $\val$. If $\idxh = |\rs'|$, $\idxt <
  |\rs_{1}|$ and a parse tree occurring in $\val$ at level $\idxt + 1$
  or lower is not $\var'$-occurrence free, then compactness of
  $\val$ implies that $\prstr_{1}$ is $\var'$-recurrence free, so
  immediate subtrees of $\prstr_{1}$ are $\var'$-occurrence free.
  The new configuration also satisfies property
  \ref{it:singleOccFree}: suppose for any $1 \le \idxo < |\rs'| $, all
  parse trees occurring in $\val'$ at level $\idxo + 1$ or lower are
  $\rp_{\idxo}.\cur$-occurrence free and $\rp_{\idxo + 1}.\folup \ne
  \emptyset$. If $\idxo < |\rs'| - 1$, then
  $\vars[ \val \downarrow \idxo + 1] = \vars [ \val_{1} \downarrow
  \idxo + 1] \subseteq \vars [ \val' \downarrow \idxo + 1]$ implies
  that all parse trees occurring in $\val$ at level $\idxo + 1$ or
  lower are $\rp_{\idxo}.\cur$-occurrence free.  By compactness of
  $\val$, $\prstr_{1}$ is
  $\rp_{\idxo}.\cur$-occurrence free. If
  $\idxo = |\rs'| - 1$, $\prlen \ge 2$ implies that $\prstr_{1}$ is
  $\rp_{\idxo}.\cur$-occurrence free by the choice of $\prstr_{1}$
  made at the beginning of this case. If $\idxo =
  |\rs'| - 1$, $\prlen
  = 1$ implies that $\rp_{\idxo + 1}.\folup = \emptyset$, so this case
  does not apply.
  Finally, the new configuration satisfies property \ref{it:recFree}:
  for some $1 \le \idxo < |\rs'|$,
  suppose there is a parse tree occurring in $\val'$ at level $\idxo +
  1$ or lower that is not $\rp_{\idxo}.\cur$-occurrence free. If
  $\idxo < |\rs_{1}|$ and a
  parse tree occurring in $\val$ at level $\idxo + 1$ or lower is not
  $\rp_{\idxo}.\cur$-occurrence free, then by compactness of
  $\val$, $\prstr_{1}$ is $\rp_{\idxo}.\cur$-recurrence
  free. If $\idxo < |\rs_{1}|$ and all
  parse trees occurring in
  $\val$ at level $\idxo + 1$ or lower are
  $\rp_{\idxo}.\cur$-occurrence free, then by compactness of
  $\val$, $\prstr_{1}$ is
  $\rp_{\idxo}.\cur$-occurrence free. If
  $\idxo = |\rs_{1}|$ and $\prlen = 1$, then $\rp_{\idxo + 1}.\folup =
  \emptyset$, so this case does not apply. If $\idxo = |\rs_{1}|$ and
  $\prlen \ge 2$, then $\prstr_{1}$ is $\rp_{\idxo}.\cur$-recurrence
  free by the choice of $\prstr_{1}$ we made at the beginning of this
  case.

  Case 4: $\rs = \rs_{1} \cdot ( \var, \Lbrack \var_{1}, \dots,
  \var_{\prlen}\Rbrack)$, $\prlen \ge 1$ and all children of the root
  of $\prstr$ are labelled by terminals. In this case, $\var
  \rightsquigarrow \word ( \prstr)$ is a production of $\cfg$. Let
  $\val ( |\rs_{1}| + 1) = \Lbrack \prstr_{1}, \dots,
  \prstr_{\prlen}\Rbrack$ such that $\prstr_{1}, \dots,
  \prstr_{\prlen}$ are rooted at $\var_{1}, \dots, \var_{\prlen}$
  respectively. If $\prlen \ge 2$ and $\rs_{1}$ ends with $\var' \in
  \vars$ such that one of the trees $\prstr_{1}, \dots,
  \prstr_{\prlen}$ is $\var'$-occurrence free, that one should be
  chosen as $\prstr'$ for the next configuration $\config'$.  If
  $\prlen \ge 2$ and $\rs_{1}$ ends with $\var' \in \vars$ such that
  none of the the trees among $\prstr_{1}, \dots, \prstr_{\prlen}$ are
  $\var'$-occurrence free, make one of them $\var'$-recurrence free by
  applying \lemref{lem:parseTreeReduceRecurrence} and choose that as
  $\prstr'$ for the next configuration $\config'$. Assume \ac{wlog}
  that if $\prlen \ge 2$ and $\rs_{1}$ ends with $\var'$, then
  $\prstr_{1}$ is $\var'$-occurrence free when possible and
  $\var'$-recurrence free otherwise.
  Let $\rs' =
  \rs_{1} \cdot (\var_{1}, \Lbrack \var_{2}, \dots, \var_{\prlen}
  \Rbrack)$. If variable $\var_{1}$ occurs more than once in
  $\rs_{1}$, it leads to a contradiction since compactness of $\val$
  implies that $\prstr_{1}$ is $\var_{1}$-occurrence free.  Hence
  $\var_{1}$ occurs at most once in $\rs_{1}$ and hence $\rs'$ is a
  state of $\aut ( \cfg)$. Let $\val_{1}$ be such that $\val_{1}
  \restr ( |\rs'|-1) = \val \restr ( |\rs'| - 1)$ and $\val_{1} (
  |\rs'|) = \Lbrack \prstr_{2}, \dots, \prstr_{\prlen} \Rbrack$. This
  new valuation $\val_{1}$ satisfies CP.\ref{it:cpOccFree} and
  CP.\ref{it:cpSingleOccFree} since it is
  obtained from $\val$ by moving subtrees within parse trees occurring
  at level $|\rs'|$ and removing a parse tree from the same level. Let
  $\val'$ be the compact valuation given by
  \lemref{lem:valCompactification} such that $\prkimg ( \val') =
  \prkimg ( \val_{1})$. We take $\rs' = \rs_{1} \cdot ( \var_{1},
  \Lbrack \var_{2}, \dots, \var_{\prlen} \Rbrack)$, $\word' =
  \word ( \prstr)$, $\val'$ is the one obtained above by applying
  \lemref{lem:valCompactification} and $\prstr' = \prstr_{1}$.

  The new configuration $\config' = (\rs', \word \word', \val',
  \prstr')$ given above satisfies property \ref{it:occFree}: let
  $\var' \in \vars$ be such that $\rp_{\idxt}.\cur = \rp_{\idxh}.\cur
  = \var'$, $1 \le \idxt < \idxh \le |\rs'|$ and $\rp_{\idxh}.\folup
  \ne \emptyset$. If
  $\idxh < |\rs'|$, then $\prstr_{1}$ is $\var'$-occurrence free by
  compactness of $\val$. If $\idxh = |\rs'|$ (in which case $\var' =
  \var_{1}$) and $\idxt =
  |\rs_{1}|$, then $\rp_{\idxh}.\folup \ne \emptyset$ implies that
  $\prlen \ge 2$ and hence by the choice of $\prstr_{1}$ made at the
  beginning of this case, $\prstr_{1}$ is $\var'$-recurrence free and
  hence immediate subtrees of $\prstr_{1}$ are $\var'$-occurrence
  free.
  If $\idxh = |\rs'|$, $\idxt < |\rs_{1}|$ and all parse trees
  occurring in $\val$ at level $\idxt +
  1$ or lower are $\var'$-occurrence free, then by compactness of
  $\val$, $\prstr_{1}$ is
  $\var'$-occurrence. If $\idxh = |\rs'|$, $\idxt < |\rs_{1}|$
  and a parse tree occurring in $\val$ at level $\idxt + 1$ or lower
  is not $\var'$-occurrence free, then by compactness of $\val$,
  $\prstr_{1}$ is $\var'$-recurrence free and
  hence immediate subtrees of $\prstr_{1}$ are $\var'$-occurrence free
  (since $\prstr_{1}$ is rooted at $\var' = \var_{1}$).
  The new
  configuration also satisfies property
  \ref{it:singleOccFree}: suppose for any $1 \le \idxo < |\rs'| $, all
  parse trees occurring in $\val'$ at level $\idxo + 1$ or lower are
  $\rp_{\idxo}.\cur$-occurrence free and $\rp_{\idxo + 1}.\folup \ne
  \emptyset$.  If $\idxo < |\rs_{1}|$, then
  $\vars[ \val \downarrow \idxo + 1] = \vars [ \val_{1} \downarrow
  \idxo + 1] \subseteq \vars [ \val' \downarrow \idxo + 1]$ implies
  that all parse trees occurring in $\val$ at level $\idxo + 1$ or
  lower are $\rp_{\idxo}.\cur$-occurrence free. By compactness of
  $\val$, $\prstr_{1}$ is
  $\rp_{\idxo}.\cur$-occurrence free and so is $\prstr_{1}$. If
  $\idxo = |\rs_{1}|$, $\prlen \ge 2$ implies that $\prstr_{1}$ is
  $\rp_{\idxo}.\cur$-occurrence free by the choice of $\prstr_{1}$
  made at the beginning of this case. If $\idxo = |\rs_{1}|$, $\prlen
  = 1$ implies that $\rp_{\idxo + 1}.\folup = \emptyset$ so this case
  does not apply.
  Finally, the new configuration satisfies property
  \ref{it:recFree}:
  for some $1 \le \idxo < |\rs'|$,
  suppose there is a parse tree occurring in $\val'$ at level $\idxo +
  1$ or lower that is not $\rp_{\idxo}.\cur$-occurrence free and
  $\rp_{\idxo + 1}.\folup \ne \emptyset$. If
  $\idxo < |\rs_{1}|$ and a
  parse tree occurring in $\val$ at level $\idxo + 1$ or lower is not
  $\rp_{\idxo}.\cur$-occurrence free, then by compactness of
  $\val$, $\prstr_{1}$ is $\rp_{\idxo}.\cur$-recurrence
  free. If $\idxo < |\rs_{1}|$ and all parse trees occurring in
  $\val$ at level $\idxo + 1$ or lower are
  $\rp_{\idxo}.\cur$-occurrence free, then by compactness of
  $\val$, $\prstr_{1}$ is
  $\rp_{\idxo}.\cur$-occurrence free. If
  $\idxo = |\rs_{1}|$ and $\prlen = 1$, then $\rp_{\idxo +
  1}.\folup = \emptyset$ so this case does not apply. If $\idxo =
  |\rs_{1}|$ and $\prlen
  \ge 2$, then $\prstr_{1}$ is $\rp_{\idxo}.\cur$-recurrence free by
  the choice of $\prstr_{1}$ we made at the beginning of this case.

  Case 5: $\rs = \rs_{1} \cdot ( \var, \Lbrack \var_{1}, \dots,
  \var_{\prlen}\Rbrack)$, $\prlen \ge 1$ and at least one child of the
  root of $\prstr$ is labelled with a variable. In this case, $\var
  \rightsquigarrow \word_{0} \var_{1}' \word_{1} \cdots \var_{\numtr}'
  \word_{\numtr}$ is a production of $\cfg$ such that $\word ( \prstr)
  = \word_{0} \cdots \word_{\numtr}$. Let $\prstr_{1}, \dots,
  \prstr_{\numtr}$ be the immediate subtrees of $\prstr$ rooted at
  $\var_{1}', \dots, \var_{\numtr}'$ respectively. If $\numtr \ge 2$
  and one of the trees $\prstr_{1}, \dots, \prstr_{\numtr}$ is
  $\var$-occurrence free, that one should be chosen as $\prstr'$ for
  the next configuration $\config'$.  If $\numtr \ge 2$ and none of
  the the trees among $\prstr_{1}, \dots, \prstr_{\numtr}$ are
  $\var$-occurrence free, make one of them $\var$-recurrence free by
  applying \lemref{lem:parseTreeReduceRecurrence} and choose that as
  $\prstr'$ for the next configuration $\config'$. Assume \ac{wlog}
  that if $\numtr \ge 2$, $\prstr_{1}$ is $\var$-occurrence free when
  possible and $\var$-recurrence free otherwise. Let $\rs' = \rs_{1}
  \cdot ( \var, \Lbrack \var_{1}, \dots, \var_{\prlen}\Rbrack) \cdot (
  \var_{1}', \Lbrack \var_{2}', \dots \var_{\numtr}' \Rbrack)$. If
  variable $\var_{1}'$ occurs more than once in $\rs_{1} \cdot ( \var,
  \Lbrack \var_{1}, \dots, \var_{\prlen}\Rbrack)$, then $\config$
  would violate property \ref{it:occFree} that $\prstr_{1}$ is
  $\var_{1}'$-occurrence free.  Hence $\var_{1}'$ occurs at most once
  in $\rs_{1} \cdot ( \var, \Lbrack \var_{1}, \dots,
  \var_{\prlen}\Rbrack)$ and hence $\rs_{1} \cdot ( \var, \Lbrack
  \var_{1}, \dots, \var_{\prlen}\Rbrack) \cdot ( \var_{1}', \Lbrack
  \var_{2}', \dots \var_{\numtr}' \Rbrack)$ is a state of $\aut (
  \cfg)$. Let $\val_{1}$ be such that $\val_{1}
  \restr ( |\rs|) = \val \restr ( |\rs|)$ and $\val_{1} ( |\rs| + 1)
  = \Lbrack \prstr_{2}, \dots, \prstr_{\numtr} \Rbrack$. This new
  valuation $\val_{1}$ satisfies CP.\ref{it:cpOccFree}: since
  $\val_{1} \restr ( |\rs|)$ already satisfies
  CP.\ref{it:cpOccFree} (due to compactness of $\val$), it is enough
  to observe that for any $\var' \in \vars$ with $\rp_{\idxt}.\cur =
  \rp_{\idxh}.\cur = \var'$, $1 \le \idxt < \idxh \le |\rs|$ and
  $\rp_{\idxh}.\folup \ne \emptyset$, we have that $\prstr_{1}, \dots,
  \prstr_{\prlen}$ are $\var'$-occurrence free since $\config$
  satisfies property \ref{it:occFree}. This new configuration also
  satisfies CP.\ref{it:cpSingleOccFree}: since
  $\val_{1} \restr ( |\rs|)$ already satisfies
  CP.\ref{it:cpSingleOccFree} (due to compactness of $\val$), it is enough
  to observe that for any $\idxo \le |\rs_{1}|$, if all parse trees
  occurring in $\val_{1}$ at level $\idxo + 1$ or lower are
  $\rp_{\idxo}.\cur$-occurrence free, then by property
  \ref{it:singleOccFree}, $\prstr$ is $\rp_{\idxo}.\cur$-occurrence
  free and hence so are $\prstr_{2}, \dots, \prstr_{\numtr}$. Let
  $\val'$ be the compact
  valuation given by \lemref{lem:valCompactification} such that
  $\prkimg ( \val') = \prkimg ( \val_{1})$. We take $\rs' = \rs_{1}
  \cdot ( \var, \Lbrack \var_{1}, \dots, \var_{\prlen}\Rbrack) \cdot (
  \var_{1}', \Lbrack \var_{2}', \dots \var_{\numtr}' \Rbrack)$,
  $\word' = \word ( \prstr)$, $\val'$ is the one obtained above by
  applying \lemref{lem:valCompactification} and $\prstr' =
  \prstr_{1}$.

  The new configuration $\config' = (\rs', \word \word', \val',
  \prstr')$ given above satisfies property \ref{it:occFree}: let
  $\var' \in \vars$ be such that $\rp_{\idxt}.\cur = \rp_{\idxh}.\cur
  = \var'$, $1 \le \idxt < \idxh \le |\rs'|$ and $\rp_{\idxh}.\folup
  \ne \emptyset$. If
  $\idxh \le |\rs|$, then $\prstr_{1}$ is $\var'$-occurrence free
  since $\config$ satisfies property \ref{it:occFree}. If $\idxh =
  |\rs'|$ (in which case $\var' = \var_{1}'$) and $\idxt = |\rs|$,
  then $\var' = \var_{1}' = \var$. $\rp_{\idxh}.\folup \ne \emptyset$ implies
  that $\numtr \ge 2$ and hence $\prstr_{1}$ is $\var'$-recurrence
  free by the choice of $\prstr_{1}$ made at the beginning of this case.
  Since $\var' = \var_{1}'$ and $\prstr_{1}$ is rooted at
  $\var_{1}'$, immediate subtrees of $\prstr_{1}$ are
  $\var'$-occurrence free. If $\idxh = |\rs'|$, $\idxt < |\rs|$ and
  there is a parse tree occurring in $\val$ at level $\idxt + 1$ or
  lower that is not $\var'$-occurrence free, then
  since $\config$ satisfies proeprty \ref{it:recFree}, $\prstr$ is
  $\var'$-recurrence free and since $\prstr_{1}$ is rooted at
  $\var_{1}' = \var'$, immediate subtrees of $\prstr_{1}$ are
  $\var'$-occurrence free. If $\idxh = |\rs'|$, $\idxt < |\rs|$ and
  all parse trees occurring in $\val$ at levels $\idxt + 1$ or lower
  $\var'$-occurrence free, then since $\config$ satisfies property
  \ref{it:singleOccFree}, $\prstr$ is $\var'$-occurrence free and
  hence so is $\prstr_{1}$.
  The new configuration also satisfies
  property
  \ref{it:singleOccFree}: suppose for any $1 \le \idxo < |\rs'| $, all
  parse trees occurring in $\val'$ at level $\idxo + 1$ or lower are
  $\rp_{\idxo}.\cur$-occurrence free and $\rp_{\idxo + 1}.\folup \ne
  \emptyset$.  If $\idxo < |\rs'| - 1$, then
  $\vars[ \val \downarrow \idxo + 1] = \vars [ \val_{1} \downarrow
  \idxo + 1] \subseteq \vars [ \val' \downarrow \idxo + 1]$ implies
  that all parse trees occurring in $\val$ at level $\idxo + 1$ or
  lower are $\rp_{\idxo}.\cur$-occurrence free.  Since $\config$
  satisfies property \ref{it:singleOccFree}, $\prstr$ is
  $\rp_{\idxo}.\cur$-occurrence free and so is $\prstr_{1}$. For
  $\idxo = |\rs'| - 1$, $\numtr \ge 2$ implies that $\prstr_{1}$ is
  $\rp_{\idxo}.\cur$-occurrence free by the choice of $\prstr_{1}$
  made at the beginning of this case. For $\idxo = |\rs| - 1$, $\numtr
  = 1$ implies that $\rp_{\idxo + 1}.\folup = \emptyset$ and hence
  this case does not apply.
  Finally, the new configuration satisfies property
  \ref{it:recFree}: for some $1 \le \idxo < |\rs'|$,
  suppose there is a parse tree occurring in $\val'$ at level $\idxo +
  1$ or lower that is not $\rp_{\idxo}.\cur$-occurrence free and
  $\rp_{\idxo + 1}.\folup \ne \emptyset$. If $\idxo < |\rs|$ and a
  parse tree occurring in $\val$ at level $\idxo + 1$ or lower is not
  $\rp_{\idxo}.\cur$-occurrence free, then since $\config$ satisfies
  property \ref{it:recFree}, $\prstr$ is $\rp_{\idxo}.\cur$-recurrence
  free and hence so is $\prstr_{1}$. If $\idxo < |\rs|$ and all parse
  trees occurring in
  $\val$ at level $\idxo + 1$ or lower are
  $\rp_{\idxo}.\cur$-occurrence free, then since $\config$ satisfies
  property \ref{it:singleOccFree}, $\prstr$ is
  $\rp_{\idxo}.\cur$-occurrence free and so is $\prstr_{1}$. If
  $\idxo = |\rs|$ and $\numtr= 1$, then $\rp_{\idxo + 1}.\folup =
  \emptyset$ so this case does not apply. If $\idxo = |\rs|$ and $\numtr
  \ge 2$, then $\prstr_{1}$ is $\rp_{\idxo}.\cur$-recurrence free by
  the choice of $\prstr_{1}$ we made at the beginning of this case.
\end{proof}

Now we are ready to prove that for every word $\word$ generated by a
\ac{CFG} $\cfg$, $\aut( \cfg)$ accepts a word $\word'$ such that $\prkimg
( \word) = \prkimg ( \word')$.
\begin{theorem}
  \label{thm:autRunExist}
  If a word $\word$ can be derived in $\cfg$ from the axiom
  $\macaxiom$, then the automaton $\aut ( \cfg)$ accepts some word
  $\word'$ such that $\prkimg ( \word') = \prkimg ( \word)$.
\end{theorem}
\begin{proof}
  Let $\config = (\rs, \word, \val, \prstr)$ be a configuration. We
  will prove by induction on $|\config|$ that there is a word
  $\word_{1}$ such that $\rs \xLongrightarrow{\word_{1}}
  (\bot,\emptyset)$ and $\prkimg( \word_{1}) = \prkimg( \val) +
  \prkimg ( \yield ( \prstr))$. For the base case $|\config| = 1$, we
  can take $\word_{1} = \epsilon$.

  For the induction step, suppose $|\config| > 1$. Let
  $\config' = (\rs', \word \word', \val', \prstr')$ be the
  configuration given by \lemref{lem:autRun}. We have $\rs
  \xLongrightarrow{\word'}\rs'$, $\prkimg(\val) +
  \prkimg ( \yield ( \prstr)) = \prkimg( \val') + \prkimg(
  \word') + \prkimg ( \yield (\prstr'))$ and $|\config'| < |\config|$.
  By induction hypothesis, there is a word $\word_{2}$ such that
  $\rs' \xLongrightarrow{\word_{2}} (\bot, \emptyset)$ and
  $\prkimg( \word_{2}) = \prkimg( \val') + \prkimg( \yield (
  \prstr'))$. Putting things together, we get $\rs
  \xLongrightarrow{ \word'} \rs' \xLongrightarrow{ \word_{2}}
  (\bot, \emptyset)$ and $\prkimg( \val) + \prkimg( \yield ( \prstr))
  = \prkimg( \word') + \prkimg( \word_{2})$. Now we can take
  $\word_{1} = \word' \word_{2}$ to complete the induction step.

  Now we will prove the lemma. Let $\prstr$ be a parse tree associated
  with the derivation of $\word$ from $\macaxiom$ so that $\yield
  ( \prstr) = \word$. Consider the configuration $(
  (\macaxiom,\emptyset), \epsilon, \{1 \to \emptyset\}, \prstr)$. From
  the above result, we get a word $\word_{1}$ such that $(\macaxiom,
  \emptyset) \xLongrightarrow{\word_{1}} (\bot, \emptyset)$ and
  $\prkimg ( \word_{1}) = \prkimg ( \{1 \to \emptyset\}) + \prkimg
  ( \yield( \prstr)) = \prkimg ( \word)$.
\end{proof}

Next we will prove the converse direction: if $\aut ( \cfg)$ accepts a
word $\word$, then $\cfg$ can generate a word $\word'$ from
$\macaxiom$ such that $\prkimg ( \word') = \prkimg( \word)$. Given a
reminder sequence $\rs$, we denote by $\Lbrack \rs \Rbrack$ the
multiset over $\vars$ such that for all $\var \in \vars$, $\Lbrack \rs
\Rbrack( \var) = \sum_{1 \le \idxo \le |\rs|}\rp_{\idxo}.\folup (
\var) + \Lbrack \rp_{|\rs|}.\cur \Rbrack ( \var)$. For any word
$\pword$ over $\termins \cup \vars$, $\pword \restr \vars$ ($\pword
\restr \termins$) is the word
obtained from $\pword$ by replacing every occurrence of an element
from $\termins$ ($\vars$) with $\epsilon$ respectively.
\begin{theorem}
  \label{thm:cfgDerExists}
  If if $\aut ( \cfg)$ accepts a word $\word$, then $\cfg$ can
  generate a word $\word'$ from $\macaxiom$ such that $\prkimg (
  \word') = \prkimg( \word)$.
\end{theorem}
\begin{proof}
  We claim that if $(\macaxiom,\emptyset) \xLongrightarrow { \word}
  \rs$, then $\cfg$ can generate a word $\pword$ such that $\Lbrack
  \rs \Rbrack = \prkimg ( \pword \restr \vars)$ and $\prkimg ( \word)
  = \prkimg ( \pword \restr \termins)$. Proof is by induction on
  length $\trlen$ of the run of $\aut ( \cfg)$ on $\word$. For the
  base case $\trlen = 0$, we can take $\pword = \macaxiom$.

  For the induction step, suppose $(\macaxiom, \emptyset)
  \xLongrightarrow{ \word_{1}'} \rs_{1} \xLongrightarrow{\word_{2}'}
  \rs$, where $\rs_{1} \xLongrightarrow{\word_{2}'} \rs$ is one of the
  transition relations used by $\aut ( \cfg)$ (from
  \defref{def:ReminderAut}). By induction hypothesis, $\cfg$ can
  generate a word $\pword'$ such that $\Lbrack
  \rs_{1} \Rbrack = \prkimg ( \pword' \restr \vars)$ and $\prkimg (
  \word_{1}')
  = \prkimg ( \pword' \restr \termins)$.

  Case 1: $\var \rightsquigarrow \word_{0} \var_{1} \word_{1} \cdots
  \var_{\prlen} \word_{\prlen}$ is a production with $\prlen \ge 1$,
  $\word_{2}' = \word_{0}\word_{1} \cdots \word_{\prlen}$ and $\rs_{1}
  = \rs_{2} \cdot (\var, \emptyset) \xLongrightarrow{ \word_{2}'}
  \rs_{2} \cdot (\var_{\idxt}, \Lbrack \var_{1}, \dots,
  \var_{\prlen}\Rbrack \ominus \Lbrack \var_{\idxt} \Rbrack) = \rs$.
  Since $\rp_{|\rs_{1}|}.\cur = \var$, $\prkimg ( \pword' \restr
  \vars) ( \var) \ge 1$. Let $\pword' = \pword_{1} \var
  \pword_{2}$. We can take $\pword = \pword_{1} \word_{0} \var_{1}
  \word_{1} \cdots \var_{\prlen} \word_{\prlen} \pword_{2}$.

  Case 2: $\var \rightsquigarrow \word_{0} \var_{1} \word_{1} \cdots
  \var_{\prlen} \word_{\prlen}$ is a production with $\prlen \ge 1$,
  $\word_{2}' = \word_{0}\word_{1} \cdots \word_{\prlen}$, $\varmults
  \ne \emptyset$ and $\rs_{1} = \rs_{2} \cdot (\var, \varmults)
  \xLongrightarrow{\word_{2}'} \rs_{2} \cdot (\var, \varmults) \cdot
  (\var_{\idxt}, \Lbrack \var_{1}, \dots, \var_{\prlen}\Rbrack \ominus
  \Lbrack \var_{\idxt} \Rbrack) = \rs$. Since $\rp_{|\rs_{1}|}.\cur =
  \var$, $\prkimg ( \pword' \restr \vars) ( \var) \ge 1$. Let $\pword'
  = \pword_{1} \var \pword_{2}$. We can take $\pword = \pword_{1}
  \word_{0} \var_{1} \word_{1} \cdots \var_{\prlen} \word_{\prlen}
  \pword_{2}$.

  Case 3: $\var \rightsquigarrow \word_{2}'$ is a production and
  $\rs_{1} = \rs_{2} \cdot (\var, \emptyset)
  \xLongrightarrow{ \word_{2}'} \rs_{2} \cdot (\bot, \emptyset) =
  \rs$. Since $\rp_{|\rs_{1}|}.\cur = \var$, $\prkimg ( \pword' \restr
  \vars) ( \var) \ge 1$. Let $\pword' = \pword_{1} \var \pword_{2}$.
  We can take $\pword = \pword_{1} \word_{2}' \pword_{2}$.

  Case 4: $\var \rightsquigarrow \word_{2}'$ is a production and
  $\rs_{1} = \rs_{2} \cdot (\var, \varmults \oplus \Lbrack \var'
  \Rbrack) \xLongrightarrow{\word_{2}'} \rs_{2} \cdot ( \var',
  \varmults) = \rs$. Since $\rp_{|\rs_{1}|}.\cur = \var$, $\prkimg (
  \pword' \restr \vars) ( \var) \ge 1$. Let $\pword' = \pword_{1} \var
  \pword_{2}$.  We can take $\pword = \pword_{1} \word_{2}'
  \pword_{2}$.

  Case 5: $\rs_{1} = \rs_{2} \cdot (\var, \varmults \oplus \Lbrack
  \var' \Rbrack) \cdot (\bot, \emptyset) \xLongrightarrow{ \epsilon}
  \rs_{2} \cdot ( \var', \varmults)$. We can take $\pword = \pword'$.
  This completes the induciton step and hence the claim is true.

  Now we will prove the lemma. Suppose $(\macaxiom, \emptyset)
  \xLongrightarrow{ \word} (\bot, \emptyset)$. By the above claim,
  $\cfg$ can generate a word $\pword$ such that $\zv = \Lbrack (\bot,
  \emptyset) \Rbrack = \prkimg ( \pword \restr \vars)$ and $\prkimg (
  \word) = \prkimg ( \pword \restr \termins)$.
\end{proof}

\paragraph{Acknowledgements} The author would like to thank Pierre
Ganty for helpful discussions.

\bibliographystyle{plain}
\bibliography{References}

\newpage
\appendix
\end{document}

%% file: macros.tex
\newcommand{\nat}{\mathbb{N}}
\newcommand{\Oh}{\mathcal{O}}
\newcommand{\poly}{\mathit{poly}}

%Acronyms
\acrodef{CFG}{Context Free Grammar}
\acrodef{cfl}[CFL]{Context Free Language}
%\acrodef{cfg}[CFG]{Context Free Grammar}
\acrodef{wlog}{without loss of generality}
\acrodef{fpt}[\textsc{Fpt}]{Fixed Parameter Tractable}

%Context free grammars
\newcommand{\regms}{d}
\newcommand{\degree}{m}
\newcommand{\cfg}{G}
\newcommand{\vars}{V}
\newcommand{\var}{A}
\newcommand{\numvars}{n}
\newcommand{\maxtermins}{e}
\newcommand{\termins}{\Sigma}
\newcommand{\letter}{\sigma}
\newcommand{\prodns}{P}
\newcommand{\macaxiom}{S}
\newcommand{\word}{w}
\newcommand{\pword}{u}
\newcommand{\prlen}{r}
\newcommand{\numtr}{s}
\newcommand{\reachrel}{\xlongrightarrow{+}}
\newcommand{\prstr}{t}
\newcommand{\prkimg}{\Pi}
\newcommand{\zv}{\overline{\mathbf{0}}}
\newcommand{\yield}{Y}
\newcommand{\height}{h}
\newcommand{\info}{f}
\newcommand{\numpv}{p}

%Graphs
\newcommand{\graph}{H}
\newcommand{\vertexs}{V}
\newcommand{\edges}{E}
\newcommand{\tree}{\mathcal{T}}
\newcommand{\bag}{B}
\newcommand{\tn}{\eta}
\newcommand{\nodes}{\mathrm{Nodes}}
\newcommand{\verto}{v}
\newcommand{\tw}{\mathit{tw}}
\newcommand{\remgr}[1]{R(#1)}
\newcommand{\troot}{\mathrm{root}}

%Automata
\newcommand{\macset}{U}
\newcommand{\elem}{u}
\newcommand{\lang}{L}
\newcommand{\mults}[1]{\mathbf{#1}}
\newcommand{\prstrmultss}{\mathbb{T}}
\newcommand{\varmults}{\mults{v}}
\newcommand{\rp}{R}
\newcommand{\rs}{\mathit{RS}}
\newcommand{\rslen}{p}
\newcommand{\cur}{\mathrm{Current}}
\newcommand{\folup}{\mathrm{Followup}}
\newcommand{\aut}{\mathcal{A}}
\newcommand{\val}{\mathit{val}}
\newcommand{\config}{c}
\newcommand{\restr}{\upharpoonright}

\newcommand{\idxo}{i}
\newcommand{\idxt}{j}
\newcommand{\idxh}{k}
\newcommand{\trlen}{\ell}

%For referring to definitions, theorems etc.
\newcommand{\Defref}[1]{Definition~\ref{#1}}
\newcommand{\defref}[1]{Def.~\ref{#1}}
\newcommand{\Thmref}[1]{Theorem~\ref{#1}}
\newcommand{\thmref}[1]{Theorem~\ref{#1}}
\newcommand{\corref}[1]{Corollary~\ref{#1}}
\newcommand{\Lemref}[1]{Lemma~\ref{#1}}
\newcommand{\lemref}[1]{Lemma~\ref{#1}}
\newcommand{\Propref}[1]{Proposition~\ref{#1}}
\newcommand{\propref}[1]{Prop.~\ref{#1}}
\newcommand{\Claimref}[1]{Claim~\ref{#1}}
\newcommand{\claimref}[1]{Claim~\ref{#1}}
\newcommand{\algoref}[1]{Algorithm~\ref{#1}}
\newcommand{\Algoref}[1]{Algorithm~\ref{#1}}
\newcommand{\Figref}[1]{Figure~\ref{#1}}
\newcommand{\figref}[1]{Fig.~\ref{#1}}
\newcommand{\Tabref}[1]{Table~\ref{#1}}
\newcommand{\tabref}[1]{Table~\ref{#1}}
\newcommand{\Secref}[1]{Section~\ref{#1}}
\newcommand{\secref}[1]{Sect.~\ref{#1}}
\newcommand{\Chref}[1]{Chapter~\ref{#1}}
\newcommand{\chref}[1]{Chap.~\ref{#1}}
\newcommand{\apndref}[1]{Appendix~\ref{#1}}
\newcommand{\Apndref}[1]{App.~\ref{#1}}

%% file: tikzmacros.tex
\newlength{\ml}
\setlength{\ml}{1cm}

%% file: ReminderGraphExample.tex
\begin{tikzpicture}[>=stealth]
  \node[circle] (a) at (0\ml,0\ml) {$\var$};
  \node[circle] (a1) at ([xshift=-2\ml,yshift=-1\ml]a)  {$\var_{1}$};
  \node[circle] (a2) at ([yshift=-1\ml]a)  {$\var_{2}$};
  \node[circle] (a3) at ([xshift=2\ml,yshift=-1\ml]a)  {$\var_{3}$};

  \draw (a1) -- ([xshift=-2\ml,yshift=-2\ml]a1.center) --
  ([xshift=2\ml,yshift=-2\ml]a1.center) -- (a1);

  \node[circle] (a4) at ([yshift=-1\ml]a1.center) {$\var'$};

  \draw (a) -- (a1);
  \draw (a) -- (a2);
  \draw (a) -- (a3);
  \draw[dotted] (a1) -- (a4);
\end{tikzpicture}